\documentclass[journal]{IEEEtran}

\usepackage{cite}
\usepackage{amsmath,amssymb,amsfonts}
\usepackage{algorithmic}
\usepackage{graphicx}
\usepackage{algorithm,algorithmic}
\usepackage{hyperref}
\hypersetup{hidelinks=true}
\usepackage{textcomp}

\def\BibTeX{{\rm B\kern-.05em{\sc i\kern-.025em b}\kern-.08em
		T\kern-.1667em\lower.7ex\hbox{E}\kern-.125emX}}

\usepackage{graphics} 
\usepackage{breqn}
\usepackage{color}
\usepackage{soul}
\usepackage{lscape}
\usepackage{bbm}
\usepackage{accents}

\usepackage{tikz}
\usetikzlibrary{positioning, calc,shadows}
\usepackage{xcolor}
\usepackage[normalem]{ulem}

\usepackage{dsfont}

\newtheorem{remark}{Remark}

\newtheorem{definition}{Definition}
\newtheorem{assumption}{Assumption}
\newtheorem{theorem}{Theorem}
\newtheorem{lemma}{Lemma}

\newcommand{\be}{\begin{equation}}
\newcommand{\ee}{\end{equation}}
\newcommand{\qed}{\hfill $\blacksquare$}

\newcommand{\calZ}{{\cal Z}}

\newcommand{\eps}{\epsilon}

\newcommand{\R}{\mathbb{R}}
\newcommand{\calU}{\mathcal{U}}
\newcommand{\bu}{\boldsymbol{u}}
\def\bmat{\begin{bmatrix}}
\def\emat{\end{bmatrix}}

\newcommand{\utilde}{\underaccent{\tilde}}

\begin{document}

\title{Scenario Approach with Post-Design Certification of User-Specified Properties 
\thanks{A preliminary version of this work appeared as a conference paper in \cite{CareCampiGaratti2025}. The present paper has been substantially revised and extended. In particular, besides extensive rewriting and several additions introduced throughout the manuscript, all the material in Sections \ref{sec:strictercertification} and \ref{sec:cost-dist}, as well as the numerical example in Section \ref{sec:general_robust_poles}, is entirely new. \\
The authors gratefully acknowledge Dr. Kristina Frizyuk for her assistance with the simulations. \\ 
Paper supported by the PRIN 2022 project 2022RRNAEX ``The Scenario Approach for Control and Non-Convex Design'' (CUP: D53D23001440006), funded by the NextGeneration EU program (Mission 4, Component 2, Investment 1.1), by the PRIN PNRR project  P2022NB77E “A data-driven cooperative framework for the management of distributed energy and water resources” (CUP: D53D23016100001), funded by the NextGeneration EU program (Mission 4, Component 2, Investment 1.1), and by the FAIR (Future Artificial Intelligence Research) project, funded by the NextGenerationEU program within the PNRR-PE-AI scheme (M4C2, Investment 1.3).}
}

\author{Algo Car{\`e}, \IEEEmembership{Member, IEEE}, Marco C. Campi, \IEEEmembership{Fellow, IEEE}, Simone Garatti, \IEEEmembership{Member, IEEE}
\thanks{A. Car\`e and M.C. Campi are with the Department of Information Engineering - University of Brescia, via Branze 38, 25123 Brescia, Italia (e-mail: [algo.care,marco.campi]@unibs.it).}
\thanks{S. Garatti is with the Dipartimento di Elettronica, Informazione e Bioingegneria - Politecnico di Milano, piazza Leonardo da Vinci 32, 20133 Milano, Italia (e-mail: simone.garatti@polimi.it).}
}

\maketitle

\begin{abstract}
The scenario approach is an established data-driven design framework that comes equipped with a powerful theory linking design \emph{complexity} to \emph{generalization} properties. In this approach, data are simultaneously used both for design and for certifying the design’s reliability, without resorting to a separate test dataset. This paper takes a step further by guaranteeing additional properties, useful in post-design usage but not considered during the design phase. To this end, we introduce a two-level framework of appropriateness: \emph{baseline appropriateness}, which guides the design process, and \emph{post-design appropriateness}, which serves as a criterion for a posteriori evaluation. We provide distribution-free upper bounds on the risk of failing to meet the post-design appropriateness; these bounds are computable without using any additional test data. Under additional assumptions, lower bounds are also derived. As part of an effort to demonstrate the usefulness of the proposed methodology, the paper presents two practical examples in $H_2$ and pole-placement problems. Moreover, a method is provided to infer comprehensive distributional knowledge of relevant performance indexes from the available dataset.  
\end{abstract}

\begin{IEEEkeywords}
Scenario approach, data-driven control design, post-design verification
\end{IEEEkeywords}

\section{Introduction}
\IEEEPARstart{T}{he} scenario approach is a well-established methodology in systems and control for data-driven design with probabilistic guarantees. The distinctive feature of the scenario approach  is that data-driven designs are certified without the need for separate test datasets.

Since its introduction in the mid-2000s, \cite{CalCam:05,CalCam:06}, the scenario approach has evolved into a versatile framework supporting robust optimization, \cite{CamGa:08,CarGarCam2015,GarCarCam2023}, optimization with constraint relaxation, \cite{GarCa2019,garatti2024noncvx}, risk-averse formulations using Conditional Value at Risk (CVaR), \cite{ramponi2018expected,arici2021theory},\cite{garatti2024noncvx}, as well as broader decision-making strategies, \cite{GarCa2019,garatti2024noncvx,campi2023compression}; the reader is referred to the book  \cite{campi2018introduction} and the review article \cite{ARC2021} for a general background.

A large body of research, \cite{Alamoetal2009,SchiFaMo:13,SchiFaFrMo2014,GrammaticoEtal:14,Alamoetal2015,MarGouLy:14,MarPraLyl:14 ,ZhangEtal:2015,PeymanAxelby2015,Darivianakis2017,ramponi2018consistency,shang2019posteriori,FALSONE2019,assif2020scenario,romao2022exact,rocchetta2023survey,wangJungers2023objective,jalota2023online,Falsone2023,wang2024data,wang2024scenario,pantazis2024priori,wangMarg2025distributed}, has contributed to the theoretical development of the scenario approach, addressing the following central question: given a scenario design based on a finite sample of observations, what is the probability that it will generalize and remain appropriate for previously unseen situations? For the scope of the present paper, two contributions are particularly relevant. The first is paper \cite{GarCa2019}, where the framework of {\em consistent decision-making} was introduced, formalizing the abstract notion of {\em appropriateness} and providing tight upper and lower bounds on the probability of inappropriateness under non-degeneracy assumptions. These bounds were further studied in an asymptotic setting in \cite{campi2023compression}. The second contribution is \cite{garatti2024noncvx}, which extended the framework by {\em removing the non-degeneracy assumption} and proved that the upper bounds remain valid in this more general setup.  

This paper further advances the scenario theory by considering settings in which a design is carried out with respect to a {\em baseline appropriateness criterion}, while additional properties of interest are assessed after the design has been determined. We term these additional properties {\em post-design appropriateness conditions}. Within this extended framework, we establish upper bounds on the probability of violating post-design appropriateness, and, under additional assumptions, derive complementary lower bounds, thereby enclosing the post-design risk within certified intervals. Importantly, all bounds are computable without resorting to any additional data besides those used during the design. These guarantees provide informative and practically relevant assessments of additional or stricter performance requirements. 
Typical situations where these new results can be applied include: 
\begin{itemize}
\item[(i)] the quantification of the probability of achieving enhanced performances beyond a guaranteed baseline. For example, the assessment of the probability that a controller meets a stricter performance requirement than those enforced at design time; 
\item[(ii)] certain design goals lead to mathematical problems that are difficult to handle. For instance, their formalization may involve non-convex or computationally complex optimization procedures. In such cases, one often resorts to surrogate or heuristic design criteria to obtain a solution, and the theory presented here can then be used to verify whether the original goals of interest are met by the solution obtained using the simplified scheme. 
\end{itemize}
In the second part of this paper, we present two examples illustrating points (i) and (ii). Furthermore, when a cost quantifies the quality of a design in relation to an uncertain environment, we show that a suitable post-design criterion enables the assessment of the full cost distribution, thus providing a comprehensive evaluation of the scenario solution. 

The remainder of the paper is organized as follows. Sections \ref{technical section}A–C develop upper bounds on post-design risk, while Section \ref{sec:strictercertification} establishes lower bounds under additional conditions. Applications to two examples in $H_2$ control and pole-placement are presented in Section \ref{sec:examples}, while Section \ref{sec:cost-dist} deals with the evaluation of the distribution of costs associated with the scenario solution. Conclusions are drawn in Section \ref{sec:conclusions}. 

\section{Problem statement and technical results}
\label{technical section}

\subsection{A scenario decision framework with two notions of appropriateness} \label{sec:frame2appr}

We first recall how the data-to-decision process is modeled within the scenario approach of \cite{GarCa2019,garatti2024noncvx}. A
list $\delta_1,\ldots,\delta_N$ of observations of an  uncertain variable $\delta$ is modeled as an independent and identically distributed (i.i.d.) sample
from a probability space $(\Delta,{\cal D}, \mathbb{P})$. Each observed $\delta_i$ is called a ``scenario'', and, borrowing a machine learning terminology, the list of scenarios $\delta_1,\ldots,\delta_N$ is sometimes called a \emph{training set}. The probability $\mathbb{P}$ is typically unknown to the user, so the framework is often called {\em agnostic}, or {\em distribution-free}.

Let $\calZ$ be a generic set, which is interpreted as the domain from which a decision
$z$ has to be chosen. From an abstract perspective, a data-driven decision scheme corresponds to a family of maps from lists of scenarios to decisions, namely
$$
M_m : \Delta^m \rightarrow \calZ, \quad m = 0, 1, 2, \ldots.
$$
When $m = 0$, the list $\delta_1,\ldots,\delta_m$ is meant to be the empty list and $M_0$ returns the decision that is made when no observations are available. We will often denote the decision returned by $M_m(\delta_1,\ldots,\delta_m)$ as $z^\ast_m$. 

A further ingredient of the standard scenario approach is a criterion of appropriateness, based on which a given decision $z \in \calZ$ is deemed appropriate or inappropriate for a specific scenario $\delta \in \Delta$. The main objective of the theory of the scenario approach is to certify the level of appropriateness of the decision $M_N(\delta_1,\ldots,\delta_N)$, where $N$ is the actual number of scenarios used to design. Importantly, this certification must be based on $\delta_1,\ldots,\delta_N$ only, i.e., without the aid of additional data. This goal is pursued in \cite{GarCa2019,garatti2024noncvx} under the requirement that the maps $M_m$ satisfy certain consistency properties, which we recall in Assumption~\ref{ass:Consistency_of_appr1} below, linking the maps $M_m$ with the notion of appropriateness. 

In this paper, we consider two distinct notions of appropriateness:
\begin{itemize}
	\item \emph{Baseline appropriateness}: this is the ``standard'' notion of appropriateness, relative to which the decision map is required to satisfy the consistency condition in Assumption~\ref{ass:Consistency_of_appr1} below;
	\item \emph{Post-design appropriateness}: this is an additional, desirable notion of appropriateness, which complements baseline appropriateness in characterizing the quality of the design for a given scenario; no assumption is made that $M_m$ is consistent with respect to this post-design appropriateness.
\end{itemize}

The baseline appropriateness plays the same role as the standard notion of appropriateness of \cite{GarCa2019,garatti2024noncvx} and it is termed here ``baseline'' to contrast it with the new, post-design appropriateness. These two notions of appropriateness represent abstractions of concrete criteria encountered in applications. For example, in control design, the baseline appropriateness may correspond to the requirement that the controller stabilizes a plant, whereas the post-design appropriateness may refer to the achievement of stability with a pre-specified margin. 

For future use, we introduce notation to denote the set of decisions that are baseline and post-design appropriate for a given scenario $\delta$. 
\medskip 
\begin{definition}[sets of baseline and post-design  appropriate decisions]
The set of decisions that are baseline appropriate for a given $\delta$ is denoted by $\calZ'_\delta$, while the set of post-design appropriate decisions is denoted by $\calZ''_\delta$.\footnote{For a given $z$, the measurability of the sets $\{\delta\in\Delta: z\in \calZ'_\delta\}$ and $\{\delta\in\Delta: z \in \calZ''_\delta\}$, as well as that of similar sets defined throughout the paper, is tacitly assumed without being stated explicitly each time.} 
\hfill$\star$
\end{definition} 
\medskip 
The requirement of baseline consistency is specified in the following assumption.
\medskip
\begin{assumption}\label{ass:Consistency_of_appr1}
The maps $M_m$ satisfy the consistency requirements as per Property 1 in \cite{garatti2024noncvx} with respect to the baseline appropriateness. That is, for every integers $m\geq 0$ and $n>0$, and every $\delta_1, \ldots, \allowbreak \delta_{m}, \delta_{m+1}, \ldots, \delta_{m+n}$, the following three conditions, called \emph{permutation invariance}, \emph{confirmation under appropriateness} and \emph{responsiveness to inappropriateness}, are satisfied:  
\begin{itemize}
	\item[-] given any permutation $(i_1,\ldots,i_{m})$ of $(1,\ldots,m)$, it holds that $M_m(\delta_1,\ldots,\delta_{m}) = M_{m}(\delta_{i_1},\ldots,\delta_{i_{m}})$; 
	\item[-]  if	$M_{m}(\delta_1,\ldots,\delta_{m}) \in \calZ'_{\delta_{m+i}}$, $\forall i \in\{1,\ldots,n\}$, then $M_{m+n}(\delta_1,\ldots,\delta_{m+n}) = 	M_{m}(\delta_1,\ldots,\delta_{m})$;
	\item[-] if $\exists \, i\in\{1,\ldots,n\}: \;\; M_{m}(\delta_1,\ldots,\delta_{m}) \notin \calZ'_{\delta_{m+i}}$, then $M_{m+n}(\delta_1,\ldots,\delta_{m+n}) \neq M_{m}(\delta_1,\ldots,\delta_{m})$. 
\end{itemize}
\hfill$\star$
\end{assumption}
Given its importance, we reiterate that no assumption is made linking $M_m$ to post-design appropriateness.
\subsection{The risk of inappropriateness - review and new goal} \label{sec:riskreview}

The notion of level of appropriateness for both baseline and post-design criteria is formalized in the following definition. 
\medskip
\begin{definition}[baseline and post-design risk] \label{def:risks}
For a given decision $z\in\calZ$, the risk of baseline inappropriateness, or {\em baseline risk}, is defined as the probability that a new $\delta$ is baseline inappropriate for that decision, i.e., 
$$
R'(z):=\mathbb{P}\{\delta\in\Delta:\, z \notin \calZ'_\delta\}; 
$$
similarly, the risk of post-design inappropriateness, or {\em post-design risk}, is the probability that a new $\delta$ is post-design inappropriate for the decision, i.e.,
$$
R''(z):=\mathbb{P}\{\delta\in\Delta:\, z \notin \calZ''_\delta\}.
$$ \hfill$\star$
\end{definition}
While the baseline risk can be dealt with through the theory of \cite{GarCa2019,garatti2024noncvx}, the main contribution of the present paper is the certification of $R''(z^\ast_N)$, the post-design risk of the decision $z^\ast_N$ designed through $M_N$. Although $R''(z^\ast_N)$ cannot be directly computed because $\mathbb{P}$ is unknown, the goal is to demonstrate that certain statistics of the data can be used to bound $R''(z^\ast_N)$ with high confidence $1\!-\!\beta$ with respect to the variability of $\delta_1,\ldots,\delta_N$. The statistics that serve to certify $R''(z^\ast_N)$  are constructed from the same dataset used for design, avoiding any loss of data. 

Towards the goal of certifying $R''(z^\ast_N)$, we find it advisable to revisit the results proven in \cite{GarCa2019} and \cite{garatti2024noncvx} concerning $R'(z^\ast_N)$ as this will allow us to introduce some notions that will be instrumental in our analysis. The first relevant definition, taken from \cite{garatti2024noncvx}, is that of support list and complexity. 
\medskip
\begin{definition}[baseline support list and baseline complexity]\label{def:complexity1}
Given a list of scenarios $\delta_1,\ldots,\delta_m$, 
a baseline support list of $M_m$ is a sub-list $\delta_{i_1},\ldots,\delta_{i_k}$, with $i_1 < i_2 < \cdots < i_k$, such that:
(a) $M_m(\delta_1,\ldots,\delta_m) = M_k (\delta_{i_1},\ldots,\delta_{i_k})$; (b)   $\delta_{i_1},\ldots,\delta_{i_k}$ is irreducible, that is, no element can be further removed from
$\delta_{i_1},\ldots,\delta_{i_k}$ while leaving the decision unchanged. \\
For a given $\delta_1,\ldots,\delta_m$, there can be more than one selection of the indexes
$i_1,i_2,\ldots,i_k$, possibly with different cardinality $k$, yielding a baseline support list. The minimal cardinality among all baseline support lists is called the {\em baseline complexity} and is denoted by ${s}^{b,\ast}_m$. \hfill$\star$
\end{definition}
\medskip 
The notion of complexity plays a key role in the analysis of the baseline risk developed in \cite{garatti2024noncvx}. In fact, Theorem~4 in that paper states that $R'(z^\ast_N)$ can be probabilistically bounded according to relation 
\begin{equation} \label{appr1_upper_bound_with_eps}
\mathbb{P}^N\{R'(z^\ast_N) \leq \epsilon({s}^{b,\ast}_N)\} \geq 1-\beta,
\end{equation}
where $\beta \in (0,1)$ is a user-chosen \emph{confidence parameter}, and the function $\epsilon(\cdot)$ is defined as follows: $\epsilon(N) = 1$ and, for $k = 0, 1,\ldots, N - 1$, $\epsilon(k) := 1 - t(k)$ where $t(k)$ is the only solution in the interval $(0, 1)$ of the equation in the $t$ variable
\begin{equation} \label{eq:ruleps}
\frac{\beta}{N}\sum_{m=k}^{N-1} {m\choose k} t^{m-k}-{N\choose k} t^{N-k} = 0.\footnote{Equtation \eqref{eq:ruleps} can be efficiently solved numerically by bisection. See Appendix B.1 of \cite{campi2023compression} for a ready-to-use MATLAB code that directly returns $\epsilon(k)$.}
\end{equation}
In view of \eqref{appr1_upper_bound_with_eps}, the baseline complexity ${s}^{b,\ast}_N$, which is a statistic of the data, can be used to obtain a certified upper bound on $R'(z^\ast_N)$. Furthermore, Theorem 2 in \cite{GarCa2019} complements this result by providing, under additional conditions, a lower bound for $R'(z^\ast_N)$. 

These were known results. The goal of this paper is to assess the risk of $z^\ast_N$ with respect to {\em post-design appropriateness}. Bounding $R''(z^\ast_N)$ is nontrivial because, in general, the maps $M_m$ do  not satisfy the consistency requirements as per Property 1 in \cite{garatti2024noncvx} in relation to post-design appropriateness: the results in \cite{garatti2024noncvx} and \cite{GarCa2019} cannot be applied and a new theoretical development is required to pursue this goal. 

\subsection{Novel results: post-design risk certification} \label{sec:novelresults}

\emph{En route} to certifying the post-design risk $R''(z^\ast_N)$, our first step consists in defining new, instrumental decision maps from lists of scenarios to an augmented decision space $\calZ^+:=\calZ\times\mathbb{N}$ ($\mathbb{N}$ denotes the set of non-negative integers), along with an instrumental notion of appropriateness for these maps. 
\medskip 
\begin{definition}[instrumental decision map]
For any $m=0,1,\ldots$ and any list $\delta_1,\ldots,\delta_m$ of scenarios, the {\em instrumental decision map} is defined as $M_m^+(\delta_1,\ldots,\delta_m):=(z_m^\ast,c^\ast_m)$, where
$z_m^\ast=M_m(\delta_1,\ldots,\delta_m)$ and $c_m^\ast=\#\big\{i \in \{1,\ldots,m\} \text{ such that }  z_m^\ast \notin \calZ''_{\delta_i}\big\}$ (the symbol ``\#'' denotes cardinality).  \hfill$\star$
\end{definition}
\medskip 
\begin{definition}[instrumentally appropriate decisions]
\label{definition - instrumentally appropriate}
For every $\delta$, the set of the {\em instrumentally appropriate decisions} for $\delta$ is defined as $\calZ^+_\delta:=(\calZ'_\delta\times \mathbb{N}) \cap (\calZ''_\delta\times \mathbb{N}).$
\hfill$\star$
\end{definition}
\medskip 
According to Definition \ref{definition - instrumentally appropriate}, an instrumental decision $(z,\ell)$ is instrumentally appropriate if and only if its first element, $z$, is both baseline and post-design appropriate. In this definition, the second element, $\ell$, does not play any role. Nevertheless, as we will see later, this second element is important to develop the theoretical framework that supports the establishment of the bounds. 

The {\em instrumental risk} of an instrumental decision is defined as follows. 
\medskip 
\begin{definition}[instrumental risk]
	For an instrumental decision $(z,\ell)\in\calZ\times\mathbb{N}$, the risk of instrumental appropriateness, or {\em instrumental risk}, is defined as the probability that a new $\delta$ is instrumentally inappropriate for that decision, i.e.,
	$$R^+(z,\ell):=\mathbb{P}\{\delta\in\Delta\,:\, (z,\ell) \notin \calZ^+_\delta\}.$$ 
	\hfill$\star$
\end{definition}	
Since, as already noticed, the element $\ell$ does not play any role in the definition of instrumental appropriateness, $R^+(z,\ell)$ only depends on $z$, and it coincides with 
	\begin{equation} \label{eq:R+(z)}
		R^+(z):=\mathbb{P}\{\delta\in\Delta\,:\, z \notin \calZ'_\delta\cap\calZ''_\delta \}.  
	\end{equation} 

The following lemma is key in the development of the theory. 
\medskip 
\begin{lemma}\label{lemma:M+consistent}
The maps $M_m^+$ satisfy the consistency requirements with respect to the {\em instrumental appropriateness} (i.e.,  the same conditions in Assumption \ref{ass:Consistency_of_appr1} hold with  $M_m$ replaced by $M_m^+$ and $\calZ'_{\delta_{m+i}}$ replaced by $\calZ^+_{\delta_{m+i}}$).
\end{lemma}
\medskip 
\textit{Proof:} We recall that $M_m^+(\delta_1,\ldots,\delta_m)=(z_m^\ast,c^\ast_m)$, where $z_m^\ast=M_m(\delta_1,\ldots,\delta_m)$ and $c^\ast_m=\#\big\{i \in \{1,\ldots,m\} \text{ such that }  z_m^\ast \notin \calZ''_{\delta_i}\big\}$. Clearly, $M^+_m$ is  {\em permutation invariant} because $M_m$ is permutation invariant by Assumption~\ref{ass:Consistency_of_appr1}, and $c^\ast_m$ does not depend on the order in which data points appear in the sample. To check {\em confirmation under appropriateness}, observe that $M^+_m(\delta_1,\ldots,\delta_m) \in \calZ^+_{\delta_{m+i}}$ implies that (a) $z^\ast_m\in\calZ'_{\delta_{m+i}}$ and (b) $z^\ast_m\in\calZ''_{\delta_{m+i}}$. Since $M_m$ satisfies {\em confirmation under appropriateness} with respect to baseline appropriateness (Assumption \ref{ass:Consistency_of_appr1}), condition (a) entails that $z^\ast_{m+n}=z^\ast_m$; moreover, this fact along with condition (b) entails that $c^\ast_{m+n}=c^\ast_m$. Regarding {\em responsiveness to inappropriateness}, observe that if $M^+_{m}(\delta_1,\ldots,\delta_m)\notin\calZ^{+}_{\delta_{m+i}}$ for some $i$, then either (a) $z^\ast_{m}\notin \calZ'_{\delta_{m+i}}$ or (b) $z^\ast_{m}\notin \calZ''_{\delta_{m+i}}$. If (a),  then $z^\ast_{m+n}\neq z^\ast_m$ because $M_m$ is responsive to baseline inappropriateness, and therefore $(z^\ast_{m+n},c^\ast_{m+n})$ differs from $(z^\ast_{m},c^\ast_{m})$ in its first component; on the other hand, if $z^\ast_{m+n} = z^\ast_m$ because only (b) occurs, then $(z^\ast_{m+n},c^\ast_{m+n})$  and $(z^\ast_{m},c^\ast_{m})$ must differ in their second components, thus concluding the proof. \qed

\medskip 
Lemma \ref{lemma:M+consistent} enables the use of the results of the scenario approach to certify the instrumental risk based on the complexity of the instrumental map. In turn, this result provides the basis to establish the post-design risk, as shown in Theorem \ref{th:upper_bound} below. 

The formal definitions of support lists and of complexity for the instrumental map are the same as in Definition \ref{def:complexity1}, where: the map $M_m$ is replaced by the instrumental map $M^+_m$; the word ``baseline'' is replaced by ``instrumental''; the symbol ``$s^{+,\ast}_m$'' is used in place of ``${s}^{b,\ast}_m$'' to denote the {\em instrumental complexity}. 
\medskip 
\begin{remark}[on computing $s^{+,\ast}_N$] \label{rmk:computings*} Any instrumental support list needs to be formed by a baseline support list augmented with the remaining $\delta_i$'s in the training set for which $z^\ast_N$ is post-design inappropriate.\footnote{Indeed, an instrumental support list must return $z_N^\ast$; by progressively removing from the instrumental support list scenarios that do not change $z_N^\ast$, if any, the process terminates with a baseline support list, so that an instrumental support list must contain a baseline support list. Augmenting the baseline support list with the remaining $\delta_i$'s in the training set for which $z^\ast_N$ is post-design inappropriate is then necessary to return the second element $c_N^\ast$ in the instrumental solution. Finally, no other scenario can be added because this would infringe the requirement of irreducibility of the instrumental support list.} Nonetheless, appending to a given baseline support list all the $\delta_i$'s from the remaining ones for which $z^\ast_N$ is post-design inappropriate surely suffices to preserve $(z^\ast_N,c^\ast_N)$ but may result in a reducible list. Therefore, constructing a minimal instrumental support list requires care. On the other hand, it is worth observing that the cardinality of any list formed by a baseline support list augmented with the $\delta_i$'s in the training set for which $z^\ast_N$ is post-design inappropriate provides an upper bound on the instrumental complexity $s^{+,\ast}_N$, and using such upper bounds in place of the actual complexity in Theorem \ref{th:upper_bound} returns a valid bound on the risk since function $\epsilon(\cdot)$ can be shown to be increasing in its argument. 
\hfill $\star$  
\end{remark}
\medskip 
We are now ready to prove our first fundamental result, which provides an upper bound on $R''(z^\ast_N)$.
\medskip 
\begin{theorem}[upper bound for post-design risk] \label{th:upper_bound}  
Let $\beta \in (0,1)$ be a confidence parameter, and define
\[
\epsilon(k) := 1 - t(k), \quad k = 0, \dots, N-1, \quad \text{and} \quad \epsilon(N) = 1,
\]
where $t(k)$ is the unique solution of \eqref{eq:ruleps} in $(0,1)$. Then, under Assumption \ref{ass:Consistency_of_appr1}, it holds for every probability $\mathbb{P}$ that
\begin{equation}\label{eq:upper_bound_on_inappr2_with_eps}
	\mathbb{P}^N\{R''(z^\ast_N) \leq \epsilon(s^{+,\ast}_N)\} \geq 1-\beta.
\end{equation}
\end{theorem}
\medskip 
\textit{Proof:} Thanks to Lemma \ref{lemma:M+consistent}, we can apply Theorem 4 in \cite{garatti2024noncvx} to the instrumental maps $M_m^+$, with instrumental appropriateness given by $\calZ^+_\delta$, which yields 
$$
\mathbb{P}^N\{R^+(z^\ast_N) \leq \epsilon({s^{+,\ast}_N})\} \geq 1-\beta.
$$
Equation \eqref{eq:upper_bound_on_inappr2_with_eps} then readily follows from observing that $R''(z^\ast_N) \leq R^+(z^\ast_N)$ in view of \eqref{eq:R+(z)}, so that 
$$
\mathbb{P}^N\{R''(z^\ast_N) \leq \epsilon({s^{+,\ast}_N})\} \geq \mathbb{P}^N\{R^+(z^\ast_N) \leq \epsilon({s^{+,\ast}_N})\}.
$$
This concludes the proof. 
\qed

\section{Tight post-design risk certification} \label{sec:strictercertification}

Lower bounds on the post-design risk can be established under additional conditions, thus providing certifications that place the risk sandwiched between a lower and an upper threshold. This topic is explored in the present section. 
 
A first situation of interest occurs when post-design appropriateness is more stringent than baseline appropriateness, that is, mathematically, $\calZ''_{\delta}\subseteq \calZ'_{\delta}$. An example of this situation is provided in Section \ref{sec:examples}-A, where post-design appropriateness corresponds to performance requirements that are more stringent than those used at the design stage in an $H_2$ control problem. In this case, we say that the two levels of appropriateness are {\em nested}, and a tight \emph{upper and lower bound} on $R''(z^\ast_N)$ can be obtained under the following non-degeneracy assumption borrowed from \cite{garatti2024noncvx}.
\medskip 
\begin{assumption}[baseline non-degeneracy]\label{cond:non-degeneracy}
 For any $m$, with probability 1,
there exists a unique choice of indexes $i_1 < i_2 < \cdots < i_k$ such that $\delta_{i_1},\ldots,\delta_{i_k}$ is a baseline support list for $\delta_1,\ldots,\delta_m$.
\hfill$\star$
\end{assumption}

\medskip 
\begin{theorem}[upper and lower bounds under baseline non-degeneracy in the nested case]
\label{th:nondegerate,nested}
Let $\beta \in (0,1)$ be a confidence parameter. 
For each $k=0,\ldots,N-1$, consider the polynomial equation in the variable $t$
\begin{equation}
	\label{eq:cvxwjrulen-1}
	{N \choose k} t^{N-k} - \frac{\beta}{2N}\sum_{i=k}^{N-1} {i\choose k} t^{i-k}-\frac{\beta}{6N}\sum_{i=N+1}^{4N}{i\choose k} t^{i-k}=0,
\end{equation}
and, for $k=N$, consider
\begin{equation}
	\label{eq:cvxwjrulen}
	1-\frac{\beta}{6N}\sum_{i=N+1}^{4N}{i\choose N} t^{i-N}=0.
\end{equation}
For each $k=0,1,\ldots,N-1$, equation \eqref{eq:cvxwjrulen-1} has exactly two solutions in $[0,+\infty)$, denoted by $\underline{t}(k)$ and $\overline{t}(k)$ with $\underline{t}(k)\leq \overline{t}(k)$. 
For $k=N$, equation \eqref{eq:cvxwjrulen} admits a unique solution in $[0,+\infty)$, denoted by $\overline{t}(N)$, and we set $\underline{t}(N) = 0$. Define
$$
\begin{cases}
	\underline{\epsilon}(k) := \max\{0, 1 - \overline{t}(k)\} \\
	\overline{\epsilon}(k) := 1 - \underline{t}(k) 
\end{cases}
, \quad k=0,\ldots,N.\footnote{Also $\underline{\epsilon}(k)$ and $\overline{\epsilon}(k)$ can be efficiently computed numerically by bisection. See Appendix B.2 of \cite{campi2023compression} for a ready-to-use MATLAB code.} 
$$
Under Assumptions \ref{ass:Consistency_of_appr1} and \ref{cond:non-degeneracy}, and assuming that $\calZ''_\delta \subseteq \calZ'_\delta$ for all $\delta$, it holds that 
\begin{equation}
\label{upper and lower bound}
\mathbb{P}^N\{\underline{\epsilon}(s^{+,\ast}_N) \leq R''(z^\ast_N) \leq \overline{\epsilon}(s^{+,\ast}_N)\} \geq 1-\beta.
\end{equation}
\end{theorem}
\medskip 
\textit{Proof:} We claim that baseline non-degeneracy implies instrumental non-degeneracy, that is, with probability one there is a unique choice of indexes $i^+_1,\ldots,i^+_{k^+}$ such that $\delta_{i^+_1},\ldots,\delta_{i^+_{k^+}}$ is an instrumental support list for $\delta_1,\ldots,\delta_m$. To see this, recall that any instrumental support list needs to be formed by a baseline support list augmented with the remaining $\delta_i$'s in the training set for which $z^\ast_N$ is post-design inappropriate (a fact already noticed in Remark \ref{rmk:computings*}). Since the indexes identifying the baseline support list are unique with probability one, the uniqueness of the indexes identifying the instrumental support list with probability one immediately follows. \\ 
Thanks to the instrumental non-degeneracy, Theorem 2 of \cite{GarCa2019} can be applied to $M^+_m$, with instrumental appropriateness given by $\calZ^+_\delta$, yielding
\begin{equation}\label{eq:uplowinstr}\mathbb{P}^N\{\underline{\epsilon}(s^{+,\ast}_N)\leq R^+(z^\ast_N)\leq \overline{\epsilon}(s^{+,\ast}_N)\}\geq 1-\beta
\end{equation}
To close the proof, note that the assumption that $\calZ''_\delta \subseteq \calZ'_\delta$ for all $\delta$ (nested case) implies that $\calZ''_\delta \cap \calZ'_\delta = \calZ''_\delta$ for all $\delta$, Therefore (see Definition \ref{def:risks} and equation \eqref{eq:R+(z)}), we have that $R^+(z)=R''(z)$ for all $z$, and \eqref{eq:uplowinstr} is equivalent to \eqref{upper and lower bound}. 
\qed

\medskip 
The bounds in Theorem \ref{th:nondegerate,nested} involve computing the zeros of the polynomials \eqref{eq:cvxwjrulen-1} and \eqref{eq:cvxwjrulen}, which offer little intuitive insight into the structure of these bounds. On the other hand, as we have seen, these bounds follow from Theorem 2 of \cite{GarCa2019}, and their expressions have been analyzed in greater detail in Proposition 8 of \cite{campi2023compression}, which provides insight into their dependencies on $N$, $\beta$, and $k$. Specifically, Proposition 8 establishes that 
\begin{eqnarray}
\overline{\epsilon}(k) & \leq &
\frac{k}{N} 
+ \frac{2\sqrt{k+1}}{N} \left( \sqrt{\ln(k+1)} + 4 \right) \nonumber \\
& & + \frac{2\sqrt{k+1}\sqrt{\ln \tfrac{1}{\beta}}}{N} 
+ \frac{\ln \tfrac{1}{\beta}}{N}, \nonumber 
\end{eqnarray}
and that 
$$\underline{\epsilon}(k) \geq \frac{k}{N} 
- \frac{3\sqrt{k+1}}{N} \left( \sqrt{\ln(k+1)} + 2 \right) 
- \frac{3\sqrt{k+1}\sqrt{\ln \tfrac{1}{\beta}}}{N}.
$$
Hence, the bounds defines intervals around $\frac{k}{N}$, with a margin that squeezes to zero as $\mathcal{O}(\sqrt{\ln(N)}/\sqrt{N})$ uniformly in $k$, while the dependence on $\beta$ is logarithmic. 
\\
 
We next move to considering the case where condition $\calZ''_\delta \subseteq \calZ'_\delta$ can be violated, still under baseline non-degeneracy. In this case, a non-vacuous lower bound can sometimes be obtained using the following theorem (an instance of use is provided in the example of Section \ref{sec:examples}-B). 
\medskip 
\begin{theorem}[upper and lower bounds under baseline non-degeneracy]
\label{th:nondegerate}
Let $\beta \in (0,1)$ be a confidence parameter, and define $\underline{\epsilon}(k),\overline{\epsilon}(k)$ as in Theorem \ref{th:nondegerate,nested} and $\epsilon(k)$ as in Theorem \ref{th:upper_bound}.
Then, under Assumptions \ref{ass:Consistency_of_appr1} and \ref{cond:non-degeneracy}, it holds that
$$
\mathbb{P}^N\{ \underline{\epsilon}(s^{+,\ast}_N)-\epsilon({s}^{b,\ast}_N)\leq R''(z^\ast_N) \leq \overline{\epsilon}(s^{+,\ast}_N)  \}\geq 1-2\beta.
$$
\end{theorem}
\medskip 
\textit{Proof:}
By applying the union bound to relation $\{\delta\in\Delta\,:\, z \notin \calZ'_\delta\cap\calZ''_\delta \} \subseteq \{\delta\in\Delta\,:\, z \notin \calZ'_\delta\} \cup \{\delta\in\Delta\,:\, z \notin \calZ''_\delta \}$, we obtain that $R^+(z) \leq R'(z) + R''(z)$; moreover, $R''(z) \leq R^+(z)$ follows from the very definition of post-design and instrumental risk. Thus, it holds that $R^+(z) - R'(z) \leq R''(z) \leq R^+(z)$ for all $z \in \calZ$, which in turn gives
\begin{equation}\label{unionbound} 
	R^+(z^\ast_N) - R'(z^\ast_N) \leq R''(z^\ast_N) \leq R^+(z^\ast_N).
\end{equation}
By Theorem 4 in \cite{garatti2024noncvx}, $R'(z^\ast_N)$ exceeds $\epsilon({s}^{b,\ast}_N)$ with probability at most $\beta$, a fact that we have already recalled in equation~\eqref{appr1_upper_bound_with_eps}. On the other hand, in the proof of Theorem \ref{th:nondegerate,nested}, we have established the validity of equation \eqref{eq:uplowinstr}, which ensures that  $R^+(z^\ast_N) \notin [\underline{\epsilon}(s^{+,\ast}_N), \overline{\epsilon}(s^{+,\ast}_N)]$ with probability at most $\beta$. Therefore, $R'(z^\ast_N) \leq \epsilon({s}^{b,\ast}_N)$ and $\underline{\epsilon}(s^{+,\ast}_N) \leq R^+(z^\ast_N)  \leq \overline{\epsilon}(s^{+,\ast}_N)$ occur simultaneously with probability no smaller than $1-2\beta$. From this, in view of \eqref{unionbound}, we conclude that relation
$$
\underline{\epsilon}(s^{+,\ast}_N)-\epsilon({s}^{b,\ast}_N) \leq R''(z^\ast_N) \leq \overline{\epsilon}(s^{+,\ast}_N)
$$
holds with probability no smaller than $1 - 2\beta$. This concludes the proof. 
\qed

\section{Examples} 
\label{sec:examples}

We demonstrate the usefulness of the developed theory through two significant control-theoretic settings: $H_2$ control and pole-placement. 

\subsection{Probabilistically robust $H_2$ control}

We consider the $H_2$ control problem for the lateral motion of an aircraft, a classic example originally introduced in \cite{tyler1966use} and later revisited in several papers and textbooks, \cite{polyak2001probabilistic,anderson2007optimal,tempo2013randomized,campi2018introduction}. The numerical values that we use are taken from \cite{tyler1966use}, while the uncertainty intervals are from \cite{campi2018introduction}. 

The lateral motion of the aircraft is described by the equation 
{\small \begin{eqnarray}
    \dot{x}(t) &  =&
  \begin{bmatrix}
0 & 1 & 0 & 0 \\
0 & L_p & L_\beta & L_r   \\
0.086 & 0 & -0.11 & -1  \\
0.086N_{\dot{\beta}} & N_p & N_\beta+0.11 N_{\dot{\beta}} & N_r
\end{bmatrix}x(t)\nonumber \\
& & + \begin{bmatrix}
 0 & 0 \\
0 & -3.91 \\
0.035 & 0  \\
-2.53 & 0.31
\end{bmatrix}u(t),\end{eqnarray}
\nonumber
}
where the four state variables are the bank angle, the derivative of the bank angle, the sideslip angle, and the yaw rate. The two inputs are the rudder deflection and the aileron deflection. The various parameters that appear in the state matrix are uncertain, i.e., the uncertain variable is $\delta = [L_p, L_\beta, L_r, N_p, N_\beta, N_{\dot{\beta}},  N_r]$, so $\Delta$ is a subset of $\mathbb{R}^7$.  For short, the state and input matrices will be written as $A_\delta$ and $B$.

Given two symmetric positive definite matrices $S$  and $R$ and an initial condition $x_0\in\mathbb{R}^4$,  we aim to design the state feedback law $u(t) = K x(t)$ that minimizes the cost function
$$
J(K,\delta) = \int_{0}^{\infty} \left[ x(t)^\top S x(t) + u(t)^\top R u(t) \right] dt,
$$
while ensuring that the closed-loop system is quadratically stable within a high-probability subset of $\Delta$.

\medskip 
\subsubsection{Review of the $H_2$ control problem with no uncertainty} 

We consider here the case in which there is no uncertainty, that is, the value of $\delta$ is fixed, and review some known facts in preparation of the probabilistically robust part. 

Given the system $\dot{x}(t)=A_\delta x(t)+B u(t)$ with $x(0)=x_0$, and a stabilizing control law $u(t)=Kx(t)$, the system state evolution is given by $x(t)=e^{(A_\delta+BK)t}x_0$, $t \! \geq \! 0$. Substituting this expression into the cost function yields  $J(K,\delta)=x_0^\top P x_0$, with $P=\int_{0}^{\infty} \left[ e^{(A_\delta+BK)^\top t}(S+K^\top R K)e^{(A_\delta+BK)t} \right]dt$, which is positive definite in view of the positive definiteness of $S$ and $R$. Observing now that $(A_\delta+BK)^\top P + P(A_\delta+BK)=\int_{0}^{\infty} \frac{d}{dt} \left[ e^{(A_\delta+BK)^\top t}(S+K^\top R K)e^{(A_\delta+BK)t} \right]dt$, which, by an application of the fundamental theorem of calculus, equals $-(S+K^\top R K)$ because the closed-loop system is stable, we conclude that $P$ satisfies the Lyapunov equation 
\begin{equation}
\label{Lyapunov equation}
(A_\delta+BK)^\top P + P(A_\delta+BK) = -(S+K^\top R K). 
\end{equation}
Moreover, from the stability of $A_\delta+BK$ easily follows that $P$ is the only solution of \eqref{Lyapunov equation}. Summarizing: a stabilizing $K$ yields $J(K,\delta)=x_0^\top Px_0$, where $P\succ 0$ is the only solution of the Lyapunov equation \eqref{Lyapunov equation}. Viceversa, given a pair of matrices $(K,P)$, with $P\succ 0$ satisfying the Lyapunov equation \eqref{Lyapunov equation}, it is straightforward to check that $V(x)=x^\top P x$ is a Lyapunov function for the closed-loop system, and therefore $K$ is stabilizing, and that $J(K,\delta) = x_0^\top P x_0$. Thus, the problem
\begin{IEEEeqnarray}{l}
	\label{LQsingledelta}
	\min_{K,P \succ 0} \quad x_0^\top P x_0  \\
	\text{subject to} \nonumber\\[2pt]
	 (A_\delta+BK)^\top P + P(A_\delta+BK)=-(S+K^\top R K) \nonumber 
\end{IEEEeqnarray}
is equivalent to minimizing $J(K,\delta)$ over all stabilizing $K$. Note also that problem \eqref{LQsingledelta} can as well be infeasible, which happens whenever $(A_\delta,B)$ is not stabilizable. 

\medskip 
\subsubsection{The scenario problem} 
We next consider the case of a variable $\delta$. 

Suppose that $\{(A_\delta,B)\}$ is quadratically stabilizable, that is, there exist a feedback gain $K$ and a matrix $P \succ 0$ such that $(A_\delta+BK)^\top P + P(A_\delta+BK) \prec 0$ for all $\delta$.\footnote{This assumption serves the purpose to make the scenario problem stated below in \eqref{scenario_problem} always feasible. While feasibility is not strictly required, and the reader can see a treatment of the case where infeasibility may occur in Section~4.2 of \cite{garatti2024noncvx}, it is here assumed to avoid complications that do not contribute to the substance of the discussion.} Next, suppose we have available a sample of $N$ i.i.d. scenarios $\delta_1,\ldots,\delta_N$ of the uncertainty parameter $\delta$, and consider the following scenario optimization problem inspired by \eqref{LQsingledelta}: 
\begin{IEEEeqnarray}{l}
	\label{scenario_problem}
\min_{K,P \succ 0} \quad x_0^\top P x_0  \\
	\text{subject to} \nonumber\\[2pt]
	(A_{\delta_i}+BK)^\top P + P(A_{\delta_i}+BK) \preceq -(S+K^\top R K), \nonumber \\ \hspace{5cm} i=1,\ldots,N. \nonumber
\end{IEEEeqnarray}
Notice that the constraint in \eqref{scenario_problem} is expressed as an inequality, a fact that makes the problem always feasible: due to quadratic stabilizability, there exists a pair $(K,P)$ that makes the left-hand side of the constraint negative definite for all $i$; then, the requirement that the left-hand side is less than or equal to $ -(S+K^\top R K)$ for all $i$ can be obtained by an appropriate rescaling of $P$. 

In special cases, it may happen that \eqref{scenario_problem} has multiple solutions. For example, if $x_0 = 0$ and $(K^\ast,P^\ast)$ is a solution, then $(K^\ast,2P^\ast)$ is also a solution. In the following, by $(K^\ast,P^\ast)$ we indicate the solution that is obtained by a tie-break rule in the domain of the pairs $(K,P)$, and refer to it as ``the solution'' of \eqref{scenario_problem}. Clearly, $K^\ast$ stabilizes all systems $(A_{\delta_i},B)$, $i=1,\ldots,N$, with $P^\ast$ being a common  Lyapunov function. Moreover, $x_0^\top P^\ast x_0$ serves as a common upper bound on the cost. To see this, rewrite the constraint $(A_{\delta_i}+BK^\ast)^\top P^\ast + P^\ast(A_{\delta_i}+BK^\ast) \preceq -(S+{K^\ast}^\top R K^\ast)$ as 
\begin{equation}
\label{Lyapunov+}
(A_{\delta_i}+BK^\ast)^\top P^\ast + P^\ast(A_{\delta_i}+BK^\ast) = -(S+{K^\ast}^\top R K^\ast) - \Xi_i, 
\end{equation}
where $\Xi_i\succ 0$ fills the gap between the two sides of the inequality. Then, we have  
\begin{align}
	\label{J(K*,delta_i)}
	& J(K^\ast,\delta_i) \nonumber\\
	&= x_0^\top \Big[\int_0^\infty
	e^{(A_{\delta_i}+BK^\ast)^\top t}
	\big(S+{K^\ast}^\top R K^\ast\big)
	e^{(A_{\delta_i}+BK^\ast)t}\,dt\Big] x_0 \nonumber\\
	&\le x_0^\top \Big[\int_0^\infty
	e^{(A_{\delta_i}+BK^\ast)^\top t}
	\big(S+{K^\ast}^\top R K^\ast \nonumber \\
	& \quad \quad +\Xi_i\big)
	e^{(A_{\delta_i}+BK^\ast)t}\,dt\Big] x_0 \nonumber\\
	&\overset{\text{(\ref{Lyapunov+})}}{=}
	-\,x_0^\top \Big[\int_0^\infty
	e^{(A_{\delta_i}+BK^\ast)^\top t}
	\big((A_{\delta_i}+BK^\ast)^\top P^\ast \nonumber \\
	& \quad \quad + P^\ast(A_{\delta_i}+BK^\ast)\big)\,
	e^{(A_{\delta_i}+BK^\ast)t}\,dt\Big] x_0 \nonumber\\
	&= -\,x_0^\top \Big[\int_0^\infty
	\frac{d}{dt}\big(
	e^{(A_{\delta_i}+BK^\ast)^\top t} P^\ast
	e^{(A_{\delta_i}+BK^\ast)t}\big)\,dt\Big] x_0 \nonumber\\
	&= x_0^\top P^\ast x_0. 
\end{align}
\subsubsection{Baseline and post-design appropriateness}
Problem \eqref{scenario_problem}, rewritten for a generic $m$ in place of $N$, defines a decision map $M_m$ from $\delta_1,\ldots,\delta_m$ to $(P^\ast,K^\ast)$. 

We say that a pair $(P,K)$ is baseline appropriate for a $\delta$ if the constraint in \eqref{scenario_problem}, rewritten for this $\delta$, is satisfied: 
$$(A_{\delta}+BK)^\top P + P(A_{\delta}+BK)\preceq -(S+{K}^\top R K).$$ 
 
As can be readily verified (see e.g. Section 3 in \cite{garatti2024noncvx} for a detailed explanation), decision maps defined through a robust optimization problem,  where constraint satisfaction plays the role of baseline appropriateness, are consistent in the sense of Assumption \ref{ass:Consistency_of_appr1}. Hence, our maps $M_m$ are consistent. 

In post-design, we consider the obtained solution $(P^\ast,K^\ast)$ to be {\em post-design} appropriate if, for a predefined value of parameter $\gamma\in(0,1)$, the {\em more stringent} constraint 
$$(A_{\delta}+B K^\ast)^\top P^\ast + P^\ast(A_{\delta}+BK^\ast)\preceq -\frac{1}{\gamma}(S+{K^\ast}^\top R K^\ast)$$  is satisfied. By a computation similar to \eqref{J(K*,delta_i)}, it is easy to see that post-design appropriateness of $(P^\ast,K^\ast)$ for $\delta$ implies that $J(K^\ast,\delta)\leq \gamma x_0^\top P^\ast x_0$. 
\medskip 
\subsubsection{Equivalent formulation as an LMI for easy implementation}
The scenario problem \eqref{scenario_problem} can be reformulated as a convex program involving LMIs (Linear Matrix Inequalities). To this end, we first note that, since $P$ is invertible, the constraint in \eqref{scenario_problem} can be rewritten as 
$$P^{-1}(A_{\delta_i}+BK)^\top  \hspace{-3pt}+ (A_{\delta_i}+BK)P^{-1}\hspace{-3pt} \preceq \hspace{-3pt} -P^{-1}(S+K^\top R K)P^{-1}.$$
Therefore, introducing the new optimization variables $Q:=P^{-1}$ and $X:=K P^{-1}$ (the same change of variables is used e.g. in Chapter 7.2.1 of \cite{boyd1994linear}), the problem can be reformulated as follows 
\begin{align} 
\min_{K,Q \succ 0} \quad & x_0^\top Q^{-1} x_0  \\
\text{subject to} &  \nonumber \\ 
&\hspace{-1cm} Q A_{\delta_i}^\top + A_{\delta_i} Q + X^\top B^\top + B X \preceq - (QSQ +X^\top R X), \nonumber \\& \hspace{4cm} i=1,\ldots,N.  \nonumber
\end{align}
Moreover, $K^\ast$ can be recovered from the solution to this problem by formula $K^\ast =  X^\ast (Q^\ast)^{-1}$.  

Now, by an application of the Schur complement (see e.g. Chapter 2 of \cite{boyd1994linear}), the problem becomes
\begin{align}
    \min_{X,Q \succ 0} \quad & x_0^\top Q^{-1} x_0 \notag  \nonumber\\
    \text{subject to} &  \nonumber \\ 
    &\hspace{-1.5cm} \begin{bmatrix}
        -(Q A_{\delta_i}^\top + A_{\delta_i} Q + X^\top B^\top + B X) & \begin{bmatrix} X^\top & Q \end{bmatrix}  \\ 
        \begin{bmatrix}X \\ Q \end{bmatrix}  & \begin{bmatrix} R^{-1} & 0 \\ 0 & S^{-1} \end{bmatrix}
    \end{bmatrix} \succeq 0,\nonumber \\ 
     & \hspace{-14pt} i=1,\ldots,N,\nonumber
\end{align}
in which the constraints are convex LMIs, and $x_0^\top Q^{-1} x_0$ is a convex cost function, as it can be argued from the following reasoning. For $Q \succ 0$, the function $x_0^\top Q^{-1}x_0$ is convex if and only if its epigraph $\{Q,h :  x_0^\top Q^{-1} x_0 \leq h\}$ is a convex set. By applying the Schur complement, the latter set can be equivalently written as
$\{Q,h : h - x_0^\top Q^{-1} x_0 \geq 0  \} = \{Q,h: \begin{bmatrix} h &  x_0^\top\\  x_0 &  Q \end{bmatrix} \succeq 0\}$, which is the convex set defined by an LMI. 
\medskip 
\subsubsection{Numerical results}
We solved the lateral motion problem with $S=0.01I$, $R=I$, $x_0=[1,1,1,1]^\top$ and $N=2000$ scenarios,\footnote{
For reproducibility, we disclose that, following \cite{campi2018introduction}, the uncertain variables were independently and uniformly sampled from the following intervals:
$L_p \in [-2.93, -1]$, $L_\beta \in [-73.14, -4.75]$, $L_r \in [0.78, 3.18]$, $N_p \in [-0.042, 0.086]$, $N_\beta \in [2.59, 8.94]$, $N_{\dot{\beta}} \in [0, 0.1]$, $N_r \in [-0.39, -0.29]$. We remark that our approach does not make use of this knowledge, which is here provided only for reproducibility purposes.} and set $\beta=10^{-7}$. The LMI reformulation of the problem was solved in MATLAB$^\copyright$ with the CVX package, \cite{cvx,gb08}, obtaining 
$$P^\ast=\begin{bmatrix}  0.2678    &0.1462  & -1.1967&    0.3250\\
    0.1462   & 0.1323 &  -1.1120&    0.3413\\
   -1.1967  & -1.1120 &  22.2723 &  -4.7196\\
    0.3250 &   0.3413  & -4.7196 &    1.7959\end{bmatrix}$$
    $$ K^\ast=\begin{bmatrix} 0.8641    & 0.9024 & -12.7202 &   4.7089\\
    0.4709  &  0.4114 &  -2.8849&    0.7778\end{bmatrix}.$$ 
This solution yields $x_0^\top P^\ast x_0=12.0367$. The cardinality of the baseline support list was ${s}^{b,\ast}_N=4$, which, according to \eqref{appr1_upper_bound_with_eps}, corresponds to a {\em baseline risk} of at most 1.53\% (this is $\epsilon({s}^{b,\ast}_N) = \epsilon(4)$) with practical certainty (confidence $1-10^{-7}$).  

Since the value $12.0367$ is relatively large, we moved to post-design to assess the probability of a smaller value. With $\gamma = 0.8$, there were $71$ scenarios  for which the more stringent constraint 
$$(A_{\delta_i}+BK)^\top P^\ast + P^\ast(A_{\delta_i}+BK^\ast)\preceq -\frac{1}{0.8}(S+{K^\ast}^\top R K^\ast)$$ 
is violated, yielding $s^{+,\ast}_N \leq 75$ (see Remark \ref{rmk:computings*}). Applying Theorem \ref{th:upper_bound}, we obtain a bound on the post-design risk of 6.94\% with confidence $1-10^{-7}$. This means that no more than 6.94\% of  the uncertain plants will result in a cost larger than $0.8\cdot 12.0367 = 9.6294$.

\subsection{Pole placement}
\label{sec:general_robust_poles}

A linear plant is described by a transfer function with uncertain parameters: 
$$
G(s, \delta) = \frac{b(s, \delta)}{a(s, \delta)} = \frac{b_0(\delta)s^n + b_1(\delta)s^{n-1} + \cdots + b_n(\delta)}{s^n + a_1(\delta)s^{n-1} + \cdots + a_n(\delta)}.
$$
We aim to design a controller 
$$
C(s) = \frac{f(s)}{g(s)} = \frac{f_1 s^{n-1} + f_2 s^{n-2} + \cdots + f_{n}}{s^n + g_1 s^{n-1} + \cdots + g_n}
$$
\begin{figure}[h]
    \centering
     \begin{tikzpicture}[auto, node distance=2cm, >=latex]
\scalebox{0.8}{
        \node[draw, circle, minimum size=10pt, inner sep=2pt] (sum) {};
        \node[draw, fill=blue!30, minimum width=2cm, minimum height=1.2cm, right=1.1cm of sum, drop shadow] (controller) {\Large $C(s)$};
        \node[draw, fill=red!30, minimum width=2cm, minimum height=1.2cm, right=1.1cm of controller, drop shadow] (plant) {\Large $G(s,\delta)$};
        \node [coordinate, right=1.1cm of plant] (fb_attach) {}; 
        \node [coordinate, right=1cm of fb_attach] (output) {};  
        \node [coordinate, below=1cm of controller] (feedback) {};
        \node [coordinate, left=1cm of sum] (input) {}; 

        \draw [->] (input) -- (sum);
        \draw [->] (sum) -- (controller);
        \draw [->] (controller) -- (plant);
        \draw [-] (plant) -- (fb_attach); 
        \draw [->] (fb_attach) -- (output); 
        \draw [->] (fb_attach) |- (feedback) -| (sum); 

        \node at ($(sum.north) + (-0.3,0.05)$) {\footnotesize $+$}; 
        \node at ($(sum.south) + (-0.2,-0.12)$) {\footnotesize $-$}; 
}
    \end{tikzpicture}
\caption{Feedback configuration for the pole placement problem.}\label{fig:closed_loop}
\end{figure}
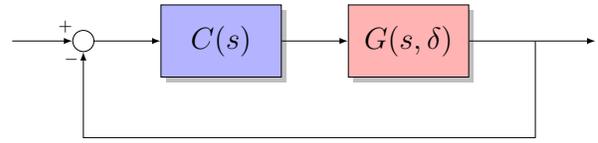
such that, with high probability, the poles of the closed-loop system in Figure \ref{fig:closed_loop}
have real part no greater than $\overline{r}$ and damping coefficient no less than $\underline{\zeta}$. This requirement defines a desirable region of the complex plane (a conic sector): 
$$ 
{\cal S}_{\overline{r},\underline{\zeta}}:=\left\{ s\in \mathbb{C}: \Re(s)\leq \overline{r}, \ \frac{\Im(s)}{\Re(s)}\leq \frac{\sqrt{1-\underline{\zeta}^2}}{\underline{\zeta}} \right\}.
$$
This objective is not easily cast into an algorithmic form because of the complex dependence of the polynomial's roots on its coefficients. Therefore, we resort to a heuristic approach: minimizing the distance between the coefficients of the closed-loop polynomial 
\begin{align}
p_{cl}(s, \delta) &= s^{2n} + p_{cl,1}(\delta)s^{2n-1 }+ \cdots + p_{cl,2n}(\delta) \nonumber\\ &= a(s, \delta) g(s) + b(s, \delta) f(s), \nonumber
\end{align}
with those of a reference polynomial 
$$
r(s) = s^{2n} + r_1 s^{2n-1} + \cdots + r_{2n} 
$$
chosen in such a way that its roots lie well within the desired region ${\cal S}_{\overline{r},\underline{\zeta}}$.

More precisely, given a list of scenarios $\delta_1, \ldots, \delta_N$, the controller coefficients are obtained from the following optimization problem: 
$$
\min_{f_1,\dots,f_{n},g_1,\dots,g_n} \left[ \sum_{j=1}^{2n} \max_{i=1,\ldots,N} | p_{cl,j}(\delta_i) - r_j | \right].
$$
Since the coefficients $p_{cl,j}(\delta)$ depend linearly on the controller parameters $f_k$ and $g_k$, this is a convex problem and can be efficiently solved. We denote by $f_1^\ast,f_2^\ast,\ldots,f_n^\ast,g_1^\ast,\ldots,g_n^\ast$ its optimal solution (after breaking possible ties in the domain of coefficients by any rule), and by  $p_{cl}^\ast(s,\delta)$ the corresponding uncertain closed-loop polynomial. 
\medskip 
\subsubsection{Equivalent formulation as a linear problem}
The following reformulation of the problem is equivalent but handier: 
\begin{IEEEeqnarray}{l}
	\label{scenario_prog_polyn}
	\min_{f_1,\dots,f_n,\; g_1,\dots,g_n,\; h_1,\dots,h_{2n}} 
	\quad \sum_{j=1}^{2n} h_j \IEEEyesnumber\\[4pt]
	\text{subject to} \nonumber\\[2pt]
	|p_{cl,j}(\delta_i)-r_j| \le h_j, \quad j=1,\dots,2n,\; i=1,\dots,N. \nonumber
\end{IEEEeqnarray}
In fact, \eqref{scenario_prog_polyn} can be easily turned into a linear program  by replacing the constraint  $|p_{cl,j}(\delta_i) - r_j| \leq h_j$ with the two linear inequalities $p_{cl,j}(\delta_i) - r_j \leq h_j$ and $-(p_{cl,j}(\delta_i) - r_j) \leq h_j$. Moreover, the solution to \eqref{scenario_prog_polyn} incorporates both the controller $f_1^\ast,\ldots,f_n^\ast,g_1^\ast,\ldots,g_n^\ast$ (as above, possible ties are broken in the domain of coefficients by any rule) and $h_1^\ast,\ldots,h_n^\ast$, which bound the maximum variability of the coefficients of $p^\ast_{cl}(s,\delta_i)$ as $a(s,\delta_i)$ and $b(s,\delta_i)$ vary across the $N$ scenarios. 
\medskip 
\subsubsection{Baseline and post-design appropriateness} Problem \eqref{scenario_prog_polyn}, rewritten for a generic $m$ in place of $N$, defines a decision map $M_m$ from $\delta_1,\ldots,\delta_m$ to $(f_1^\ast,\ldots,f_n^\ast,g_1^\ast,\ldots, g_n^\ast,h_1^\ast,\ldots, h_{2n}^\ast)$. 

We say that a given  $(f_1,\ldots,f_n,g_1,\ldots, g_n,h_1,\ldots, h_{2n})$ is baseline appropriate for a $\delta$ if 
$$
|p_{cl,j}(\delta)-r_j|\leq h_j \mbox{ for all } j=1\,\ldots,2n, 
$$ 
and maps $M_m$ are easily seen to be consistent. 

While baseline appropriateness serves as a tool to shift the closed-loop poles toward a location well inside the desired conic sector ${\cal S}_{\overline{r},\underline{\zeta}}$, its satisfaction does not imply that the closed-loop system poles actually lie within ${\cal S}_{\overline{r},\underline{\zeta}}$. Then, in  post-design, $(f_1^\ast,\ldots,f_n^\ast,g_1^\ast,\ldots, g_n^\ast,h_1^\ast,\ldots, h_{2n}^\ast)$ is deemed \emph{post-design} appropriate if the roots of $p_{cl}^\ast(s,\delta_i)$ indeed belong to the conic sector ${\cal S}_{\overline{r},\underline{\zeta}}$, and guarantees are derived on the probability that the roots of $p_{cl}^\ast(s,\delta)$ belong to this region using Theorem \ref{th:upper_bound}. 

\medskip 
\subsubsection{Numerical results}
The pole-placement methodology is applied to the stabilization of an inverted pendulum using a linearization procedure. 

We consider a classical pendulum system (see Figure \ref{fig:pendulum}) governed by the differential equation
$$
M \ell\ddot{\theta}(t) = - \alpha \dot{\theta}(t) - Mg  \sin(\theta(t)) +  u(t),
$$
where $\theta$ is the angular displacement measured counterclockwise from the downward vertical, $u$ is a tangential force (control variable), $g$ is the gravitational acceleration ($9.8 [m/s^2]$), $M$ the pendulum mass, $\ell$ is the rod length, and $\alpha$ is a viscous friction coefficient. The parameters $M$, $\ell$ and $\alpha$ are uncertain, and are collected in $\delta=[M,\ell,\alpha]^\top \in \mathbb{R}^3$.\footnote{For completeness, we specify the distribution of the uncertain parameters used in the simulation. $M$ is uniformly distributed between $9[Kg]$ and $10[Kg]$; $\ell$ is independent of $M$ and uniformly distributed between $0.9[m]$ and $1[m]$; the parameter $\alpha$ depends on $M$ and its conditional distribution (expressed in $[Kg \ m/s^2])$ is uniform between the values $M$ and $1.1M$. We remark that the theory developed in this paper certifies appropriateness without using this information, which is provided here solely for reproducibility purposes.}
\begin{figure}[h!]
\begin{center}
\begin{tikzpicture}
    \coordinate (O) at (0,0); 
    \coordinate (M) at (2,-2.5); 
    \coordinate (V) at (0,-3); 

    \draw[dashed] (O) -- (V);

    \draw[thick] (O) -- (M) node[midway, left] {$\ell$};

    \filldraw[black] (O) circle (2pt);

    \filldraw[gray] (M) circle (5pt) node[right=5pt] {$M$};

    \draw[->, thick] (M) --++ (0,-1.2) node[midway, right] {$Mg$};

    \draw[->, thick, red] (M) --++ (0.8,0.8) node[midway, above] {$u$};

    \draw[thick] (0,-0.8) arc[start angle=-90, end angle=-52, radius=0.8];
    \node at (0.15,-0.45) {$\theta$};

\end{tikzpicture}
\end{center}
\caption{A classic pendulum}
\label{fig:pendulum}
\end{figure}
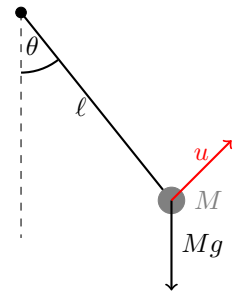

By linearizing the differential equation of the pendulum around  the upright equilibrium ($\theta = \pi)$, we find the transfer function between the control action $u$ and the angle $\theta$ to be:
$$
\frac{\Theta(s)}{U(s)}=\frac{1}{ s^2 + \frac{\alpha}{M\ell}s - \frac{g}{\ell}}.
$$
We desire closed-loop system poles with real part no greater than $\overline{r}=-0.7$ and damping coefficient no less than $\underline{\zeta}=0.5$. The procedure in Section \ref{sec:general_robust_poles} was applied with reference polynomial $r(s)=s^4+8s^3+32s^2+48s+ 36$ (which has roots in  $-1\pm j$ and $-3\pm 3 j$). The confidence parameter was set to value $\beta=10^{-5}$ and $N=2000$ independent scenarios were drawn from the uncertain parameter region. The optimization program \eqref{scenario_prog_polyn} was solved in MATLAB$^\copyright$ using the \texttt{linprog} function, which yielded the controller 
$$
C^\ast(s)=\frac{724 s + 3536}{s^2 + 6.889 s + 34.72}.
$$
The closed-loop polynomials $p_{cl}^\ast(s,\delta_i)$ corresponding to the $N=2000$ scenarios were  in the range  
\begin{equation}
s^4  +  (8\pm0.1)s^3 +  (32\pm0.5) s^2  +  (48\pm 8)s  +  (36\pm22.4),
\label{closed_loop_polyn_solution}    
\end{equation}
where the intervals of variability of the coefficients were obtained from the optimal values $h_j^\ast$. 
  
An examination of the Lagrange multipliers revealed that eight constraints were active at the optimum, and the scenarios associated with these constraints formed the unique baseline support list; hence, ${s}^{b,\ast}_N=8$.\footnote{In convex optimization problem as \eqref{scenario_prog_polyn}, any support list is always formed by scenarios corresponding to active constraints; in the instance at hand none of these constraints could be removed without altering the solution.}  A certificate (valid with confidence $1-10^{-5}$) on the {\em baseline risk} was obtained by applying \eqref{appr1_upper_bound_with_eps}: the closed-loop polynomial for a newly sampled $\delta$ falls outside the computed range with probability at most 1.64\%. Ultimately, our goal was to certify {\em post-design appropriateness}. To this end, we needed to compute the post-design complexity $s^{+,\ast}_N$. By evaluating how many of the $N=2000$ scenarios, excluding those corresponding to active constraints, led to closed-loop system poles not all contained in ${\cal S}_{-0.7,0.5}$, and incrementing ${s}^{b,\ast}_N=8$ with this number, we obtained $s^{+,\ast}_N = 199$ (see Figure~\ref{fig:swan_poles} for a plot of the closed-loop system poles for the $N=2000$ scenarios considered). 
\begin{figure}
    \centering
\includegraphics[width=0.85\linewidth]{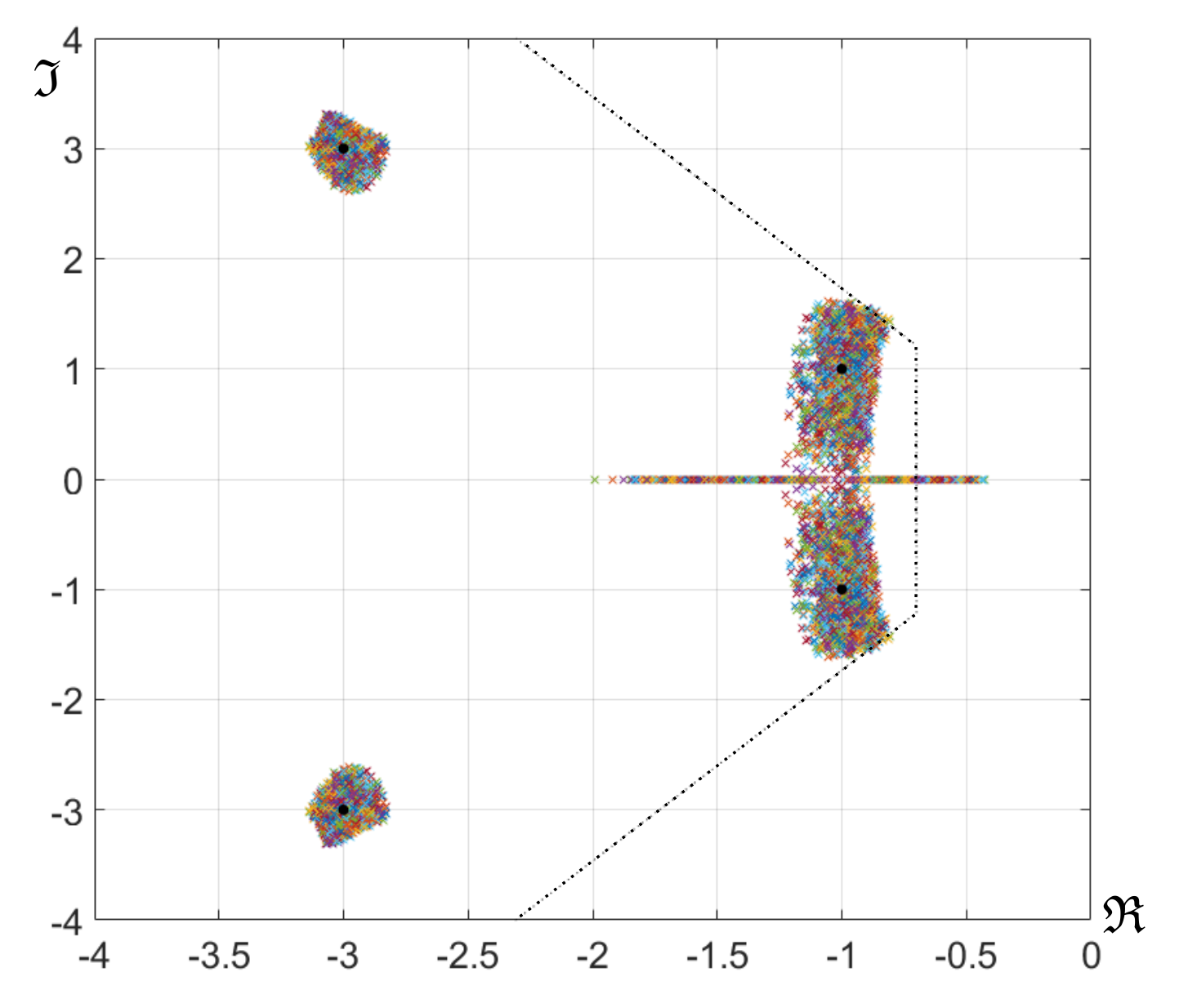}
\caption{Locations in the complex plane of the closed-loop system poles corresponding to the $N=2000$ scenarios in problem \eqref{scenario_prog_polyn}. The dotted line is the border of the desired conic sector ${\cal S}_{-0.7,0.5}$. The black dots in $-1\pm j$ and $-3\pm 3 j$ are the poles of the reference polynomial.}
    \label{fig:swan_poles}
\end{figure} 
Applying Theorem \ref{th:upper_bound}, again with $\beta=10^{-5}$, the {\em post-design risk} was found to be bounded by $\overline{\epsilon}(193) = 13.91\%$. 

In the context of  \eqref{scenario_prog_polyn}, the non-degeneracy condition in Assumption \ref{cond:non-degeneracy} is a relatively mild condition; it is equivalent to requiring that the scenarios corresponding to  active constraints together form a support list with probability 1. In this case,  Theorem \ref{th:nondegerate} allows one to certify (with confidence $1-2\cdot  10^{-5}$) that the post-design risk is between $5.07\%$ ($\underline{\epsilon}(s^{+,\ast}_N)-\epsilon({s}^{b,\ast}_N)=6.71\%-1.64\%$) and $13.56\%$   ($\overline{\epsilon}(s^{+,\ast}_N)$). 

Interestingly, the post-design evaluation can be performed for alternative regions, which also provides a way to assess how much the performance can be further guaranteed  by slightly relaxing the initial requirements. In this example, we repeated the evaluation described above reducing the requirement on the real part to $\overline{r} = -0.6, -0.5, -0.4$ while keeping $\underline{\zeta}$ at its original value $0.5$. For $\overline{r} = -0.6$, we obtained $s^{+,\ast}_N = 130$, resulting in a post-design risk between $2.32\%$ and $9.85\%$; for $\overline{r} = -0.5$, $s^{+,\ast}_N$ was $46$ with a post-design risk between $0$ and $4.66\%$; and for $\overline{r} = -0.4$, $s^{+,\ast}_N$ was $14$ with a post-design risk between $0$ and $2.17\%$. Note, however, that shifting $\overline{r}$ from $-0.7$ to $-0.4$ corresponds to an approximate $75\%$ increase in the settling time.

\section{Using post-design certifications to evaluate cost distributions} \label{sec:cost-dist}

In this section, we show how post-design risk results can be used to obtain a distributional evaluation of the scenario solution $z^\ast_N$ when performance is expressed in terms of a cost. 

Suppose that a cost function $f(z,\delta) : \calZ \times \Delta \to \R$ measures the performance of a decision $z$ when scenario $\delta$ occurs. Then, a comprehensive assessment of the quality of $z$ is provided by the cumulative distribution function (CDF) of the random variable $f(z,\delta)$: 
$$
F_z(\ell) := \mathbb{P}\{ f(z,\delta) \le \ell \}. 
$$
In the context of scenario-based decision making, one is interested in evaluating $F_{z^\ast_N}(\ell)$, the CDF associated with the scenario solution $z^\ast_N$, constructed according to a given decision scheme and baseline appropriateness criterion.\footnote{Sometimes, one takes $f(z,\delta) \le \bar{\ell}$ as baseline appropriateness, where $\bar{\ell}$ is a cost threshold relevant for design purposes.} To this end, consider a grid $\ell_1 < \ell_2 < \cdots < \ell_h$ covering a relevant interval of $\R$ over which $f(z^\ast_N,\delta)$ is expected to vary, and introduce a family of post-design appropriateness criteria
$$
\calZ''_{\delta,\ell_j} := \{ z : f(z,\delta) \le \ell_j \}, \qquad j=1,\ldots,h. 
$$
By definition of post-design risk, we have 
\begin{align}
	R''_{\ell_j}(z) := &  \; \mathbb{P}\{ \delta\in\Delta : z \notin \calZ''_{\delta,\ell_j} \} \nonumber \\
	= & \; \mathbb{P}\{ f(z,\delta) > \ell_j \} 
	= 1 - F_z(\ell_j). \label{eq:R''=1-F} 
\end{align}
Applying Theorem~\ref{th:upper_bound} for each $\ell_j$ yields that 
$R''_{\ell_j}(z^\ast_N) \le \epsilon(s^{+,\ast}_{N,j})$ with confidence $1-\beta$, where $s^{+,\ast}_{N,j}$ denotes the post-design complexity associated with $\calZ''_{\delta,\ell_j}$. Combining this result with~\eqref{eq:R''=1-F}, and using a union bound over the grid points $\ell_j$'s, gives that the inequalities 
$$
F_{z^\ast_N}(\ell_j) \;\ge\; 1 - \epsilon(s^{+,\ast}_{N,j}),
\quad j=1,\ldots,h,
$$
hold \emph{simultaneously} with confidence $1 - h\beta$.

Observe now that any lower bound on $F_{z^\ast_N}(\ell_j)$ extends to all $\ell \ge \ell_j$ because CDFs are monotonically increasing. Therefore, defining
$$
\underline{F_\epsilon}(\ell)
=
\begin{cases}
	1 - \epsilon(s^{+,\ast}_{N,h}), 
	& \ell \ge \ell_h,\\
	1 - \epsilon(s^{+,\ast}_{N,j}), 
	& \ell_j \le \ell < \ell_{j+1},\quad j=1,\ldots,h-1,\\
	0, 
	& \ell < \ell_1,
\end{cases}
$$
we obtain the bound
$$
F_{z^\ast_N}(\ell)
\;\ge\;
\underline{F_\epsilon}(\ell), 
\; \forall \ell\in\R, \quad \text{with confidence } 1 - h\beta.
$$
Thus $\underline{F_\epsilon}(\ell)$ is, with confidence $1-h\beta$, a valid lower bound for the \emph{entire} CDF of $f(z^\ast_N,\delta)$.\footnote{Note that $\underline{F_\epsilon}$ is also monotonically increasing. Indeed, since $\ell_j < \ell_{j+1}$, the implication $f(z^\ast_N,\delta) > \ell_{j+1} \Rightarrow f(z^\ast_N,\delta) > \ell_j$ holds, from which it follows that $s^{+,\ast}_{N,j} \ge s^{+,\ast}_{N,j+1}$. The monotonicity of $\underline{F_\epsilon}(\ell)$ then follows from the fact that $\epsilon(\cdot)$ is an increasing function, as previously noted in Remark~\ref{rmk:computings*}.}

\medskip
The above results can be further strengthened under Assumption~\ref{cond:non-degeneracy}. In this case, repeating the argument above, with the obvious modifications, and resorting to Theorems~\ref{th:nondegerate,nested} and~\ref{th:nondegerate} gives that
$$
1 - \utilde{\epsilon}(s^{b,\ast}_N,s^{+,\ast}_{N,j})
\;\ge\;
F_{z^\ast_N}(\ell_j)
\;\ge\;
1 - \overline{\epsilon}(s^{+,\ast}_{N,j}),
\quad j=1,\ldots,h,
$$
hold simultaneously with confidence $1-(h+r)\beta$, where
$$
\utilde{\epsilon}(s^{b,\ast}_N,s^{+,\ast}_{N,j}) =
\begin{cases}
	\underline{\epsilon}(s^{+,\ast}_{N,j}), 
	& \text{if } \calZ''_{\delta,\ell_j} \subseteq \calZ'_\delta \mbox{ for all } \delta,\\
	\underline{\epsilon}(s^{+,\ast}_{N,j}) - \epsilon(s^{b,\ast}_N), 
	& \text{otherwise},
\end{cases}
$$
and $r$ counts the number of thresholds $\ell_j$ for which the non-nested case occurs.\footnote{When baseline appropriateness is expressed as $f(z,\delta)\le\bar{\ell}$, the nested case holds for $\ell_j\le\bar{\ell}$ and fails when $\ell_j>\bar{\ell}$.}

Proceeding as before, and using the fact that an upper bound for $F_{z^\ast_N}(\ell_j)$ extends to all $\ell\le \ell_j$, we obtain upper and lower bounds on the entire CDF of $f(z^\ast_N,\delta)$. Specifically, define $\underline{F_{\eps}}$ as before but with $\overline{\epsilon}(\cdot)$ in place of $\epsilon(\cdot)$, and
$$
\overline{F_{\eps}}(\ell)
=
\begin{cases}
	1, & \ell \ge \ell_h,\\
	1 - \utilde{\epsilon}(s^{b,\ast}_N,s^{+,\ast}_{N,j}), 
	& \ell_{j-1} < \ell \le \ell_j,\quad j=2,\ldots,h,\\
	1 - \utilde{\epsilon}(s^{b,\ast}_N,s^{+,\ast}_{N,1}), 
	& \ell \le \ell_1.
\end{cases}
$$
Then,
\begin{align*}
& \overline{F_{\eps}}(\ell) \;\ge\; F_{z^\ast_N}(\ell) \;\ge\; \underline{F_{\eps}}(\ell), \quad \forall \ell\in\R, \\
& \text{with confidence } 1-(h+r)\beta.
\end{align*}
Hence, the entire CDF of $f(z^\ast_N,\delta)$ is sandwiched between two data-driven functions, resulting in a complete characterization of the performance of $z^\ast_N$ in the uncertain environment. 

\subsection{An example: open-loop input design}

The previously discussed results are now illustrated by means of a simple open-loop input design problem inspired by an example in~\cite{Cam_Gar_Ram:18}.

Consider the discrete-time uncertain linear system
\begin{equation}
	\label{random_sys}
	\eta(t+1) = A \eta(t) + B u(t), \quad\quad \eta(0) = \eta_0,
\end{equation}
where $\eta(t) \in \R^2$ is the state and $u(t) \in \R$ is the control input, which takes value in $\calU := [-10,10]$ due to actuation constraints. The initial state and the matrix $B$ are fixed and given by $\eta_0 = \bmat 1 & 1 \emat^\top$ and $B = \bmat 0 & 0.25 \emat^\top$, respectively. Instead, the state matrix $A \in \R^{2\times 2}$ is uncertain, with entries drawn independently from four Gaussian distributions with means
\begin{displaymath}
	\bar{A} =
	\bmat 0.8 & -1 \\ 0 & -0.9 \emat,
\end{displaymath}
and standard deviations $0.05(1+2v_i)$, $i=1,2,3,4$, where the $v_i$'s are independent draws of a Bernoulli random variable taking value $1$ with probability $0.03$ and value $0$ otherwise. These distributions are specified solely to enable reproducibility and they play no role in the design of the solution or in its quality assessment. The only information used for design consists of $N=1000$ realizations of $A$, say $A_1,\ldots,A_{1000}$, which are identified with the scenarios $\delta_1,\ldots,\delta_{1000}$. 

Informally, the control objective is to select the input sequence $u(0), \cdots, u(T \! - \! 1)$ so as to drive the system state at time $T=5$ as close as possible to the origin, where the distance is measured according to
the maximum norm $\left\| \eta(T) \right\|_\infty := \max(|\eta_1(T)|,|\eta_2(T)|)$. Observing that $\eta(T) = A^T \eta_0 + \sum_{t=0}^{T-1} A^{T-1-t} B u(t)$, it is immediately seen that $\left\| \eta(T) \right\|_\infty = \left\| A^T \eta_0 + R \bu \right\|_\infty$, where $R = \bmat B & A B & \cdots & {A}^{T-1} B \emat$
and $\bu = \allowbreak  \bmat u(T-1) & u(T-2) & \cdots & u(0) \emat^\top$. 

In formal terms, the input design problem is formulated as follows: 
\begin{eqnarray} \label{eq:sc-opt-relax-example}
	\min_{\stackrel{h \geq 0, \bu \in \calU^T,}{\xi_i \geq 0}} & & h +
	\rho \sum_{i=1}^N \xi_i \\
	\textrm{\rm subject to:} & & \left\| A_i^T \eta_0 + R_i \bu
	\right\|_\infty - h \leq \xi_i, \quad i = 1,\ldots,N, \nonumber
\end{eqnarray}
which is a scenario program with constraint relaxation, \cite{garatti2024noncvx}. When $\rho$ is very large, this problem reduces to 
\begin{eqnarray} 
	\min_{h \geq 0, \bu \in \calU^T} & & h \nonumber \\
	\textrm{\rm subject to:} & & \left\| A_i^T \eta_0 + R_i \bu
	\right\|_\infty \leq h, \quad i = 1,\ldots,N, \nonumber
\end{eqnarray}
which is a robust problem whose optimal value $h^\ast$ represents the maximum deviation of the final states, across the available scenarios, from the origin. Allowing for smaller values of $\rho$ results instead in a program where violation of the constraints $\left\| A_i^T \eta_0 + R_i \bu \right\|_\infty \leq h$ is tolerated, but this is discouraged by the penalization term $\rho \sum_{i=1}^N \xi_i$. This is expected to reduce the influence of particularly adverse scenarios, so yielding a better threshold $h^\ast$ and $\eta(T)$ closer to the origin for all the other scenarios. 

For each $\rho$, replacing $N$ in \eqref{eq:sc-opt-relax-example} with $m=0,1,\ldots$ defines a family of decision maps $M_m$ that are consistent with respect to the appropriateness criterion $\left\| A_i^T \eta_0 + R_i \bu \right\|_\infty - h \leq 0$, taken here as \emph{baseline appropriateness}. See \cite{garatti2024noncvx} for details. 

Problem \eqref{eq:sc-opt-relax-example} was solved for $\rho = 1$, resulting in $\bu^\ast_N = [-0.27 \ -0.58 \ -1.34 \ 0.45 \ 5.85]^\top$, $h^\ast_N = 0.15$, and this solution was found to coincide with the robust solution as no constraints were violated (we had $\xi_i = 0 $ for all $i$). We then solved \eqref{eq:sc-opt-relax-example} with $\rho = 0.05$, resulting in $\bu^\ast_N = [1.19 \ -0.38 \ -4.76 \ 0.11 \ 7.38]^\top$, $h^\ast_N = 0.07$, a less conservative solution which however violates $18$ constraints (we had $\xi_i > 0$ in $18$ cases). 

To gain further insight into the performance of these two solutions, we evaluated the distribution of $f(z^\ast_N,\delta) = \left\| A^T \eta_0 + R \bu^\ast_N \right\|_\infty - h^\ast_N$ following the methodology of Section~\ref{sec:cost-dist}. To this end, $\ell_1,\ldots,\ell_h$ were taken as a grid of $h = 100$ uniformly spaced values in the interval $[-0.15,0]$, $\beta$ was set to $10^{-7}$, and $\overline{F_{\eps}}(\ell)$ and $\underline{F_{\eps}}(\ell)$ were computed for both solutions. With high confidence $1-10^{-5}$ (this is $1-h \cdot 10^{-7}$), these functions provide an upper and a lower bound for the entire CDF of $\left\| A^T \eta_0 + R \bu^\ast_N \right\|_\infty - h^\ast_N$.\footnote{Provably, Assumption \ref{cond:non-degeneracy} is satisfied in the present case because $\delta$ admits a density. Moreover, since $\ell_j \leq 0$ for all $j$'s, the nested case applies throughout.} 

A more meaningful interpretation of the results is obtained by observing that shifting $f(z^\ast_N,\delta)$ by $h^\ast_N$ yields $\left\| A^T \eta_0 + R \bu^\ast_N \right\|_\infty$, which is the maximum norm  $\| \eta(T) \|_\infty$, of the final state. Therefore, $\overline{F_{\eps}}(\ell-h^\ast_N)$ and $\underline{F_{\eps}}(\ell-h^\ast_N)$ serve as an upper and a lower bound on the entire CDF of $\| \eta(T) \|_\infty$.

Functions $\overline{F_{\eps}}(\ell-h^\ast_N)$ and $\underline{F_{\eps}}(\ell-h^\ast_N)$ are shown in Figure~\ref{fig:f_dist_combined} for $\rho = 1$  (top panel) and $\rho = 0.05$ (bottom panel).
\begin{figure}[ht]
	\centering
	\includegraphics[width=0.85\columnwidth]{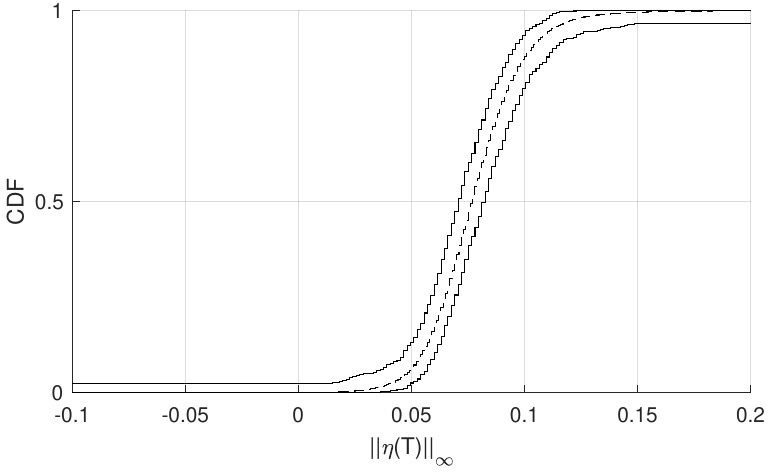}
	\\[1em] 
	\includegraphics[width=0.85\columnwidth]{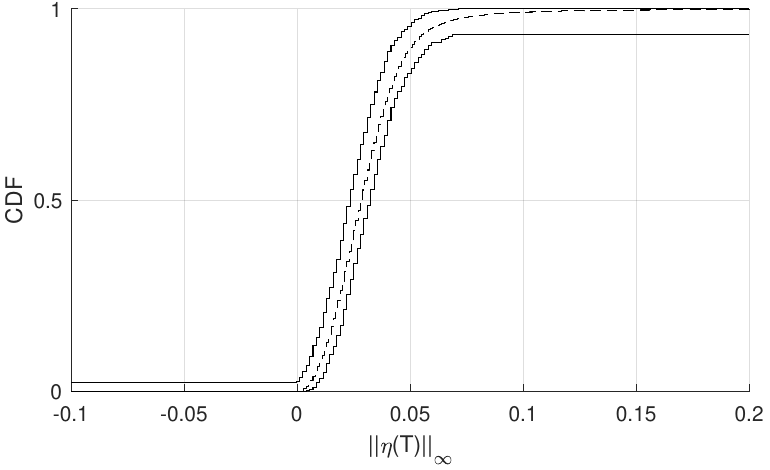}
	\caption{$\overline{F_{\eps}}(\ell-h^\ast_N)$ and $\underline{F_{\eps}}(\ell-h^\ast_N)$ (solid lines) vs. actual cumulative distribution function of $\| \eta(T) \|_\infty$ (dashed lines) for $\rho = 1$ (top) and $\rho = 0.05$ (bottom).}
	\label{fig:f_dist_combined}
\end{figure}
These plots offer several useful insights. In particular, the conservatism of the solution obtained with $\rho = 1$ is evident: in an effort to safeguard against the worst, the adopted control actions produce values of $\| \eta(T) \|_\infty$ that are predominantly distributed above $0.05$. In contrast, the solution corresponding to $\rho = 0.05$ yields values of $\| \eta(T)\|_\infty$ that are concentrated between $0$ and $0.05$, at the cost of a potentially heavier tail in the distribution.

Since this example was conducted entirely {\em in silico}, we also validated the analysis by generating $100000$ new scenarios $A_i$ and computing Monte Carlo the actual CDF of $\left\| \eta(T) \right\|_\infty$ for both solutions. The resulting CDFs, shown as dashed lines in the plots, indeed lie between $\overline{F_{\eps}}(\ell-h^\ast_N)$ and $\underline{F_{\eps}}(\ell-h^\ast_N)$ in both cases, as 
predicted with confidence $1-10^{-5}$ by the theory. 	

\section{Conclusions} 
\label{sec:conclusions}
In this paper, we have introduced a novel framework to certify data-driven decisions in post-design. We have distinguished two levels of appropriateness: \emph{baseline appropriateness}, which guides the design process, and \emph{post-design appropriateness}, which serves as a criterion for \emph{a posteriori} evaluation. These notions correspond to two different kinds of risk:  the first, termed \emph{baseline risk}, is the traditional subject of the scenario approach, while the second, termed \emph{post-design risk}, was not previously covered by the scenario theory. In alternative approaches, the post-design risk is typically evaluated using test datasets. We have shown that effective bounds on the post-design risk can be derived via an extended notion of complexity, computable directly from the training set without the need for any additional test data. Moreover, we have demonstrated that, in relevant settings, both upper and lower bound on the post-design risk can be obtained,  yielding tight evaluations. 

The theory developed in this paper is relevant when additional, or more stringent, requirements need to be verified in post-design. Moreover, it offers a viable approach to using heuristics in problems that are impractical to handle according to their original formulation, while still providing rigorous guarantees on the original requirements. These two paradigms have been illustrated by means of two examples: the design of lateral motion control, and pole-placement for stabilizing an inverted pendulum. We have also illustrated a method that exploits the notion of post-design appropriateness to infer comprehensive distributional knowledge of relevant performance indexes from the available dataset. 

\bibliographystyle{IEEEtran}
\bibliography{post-design-CCG.bib}

\begin{IEEEbiography}[{\includegraphics[width=1in,height=1.25in,clip,keepaspectratio]{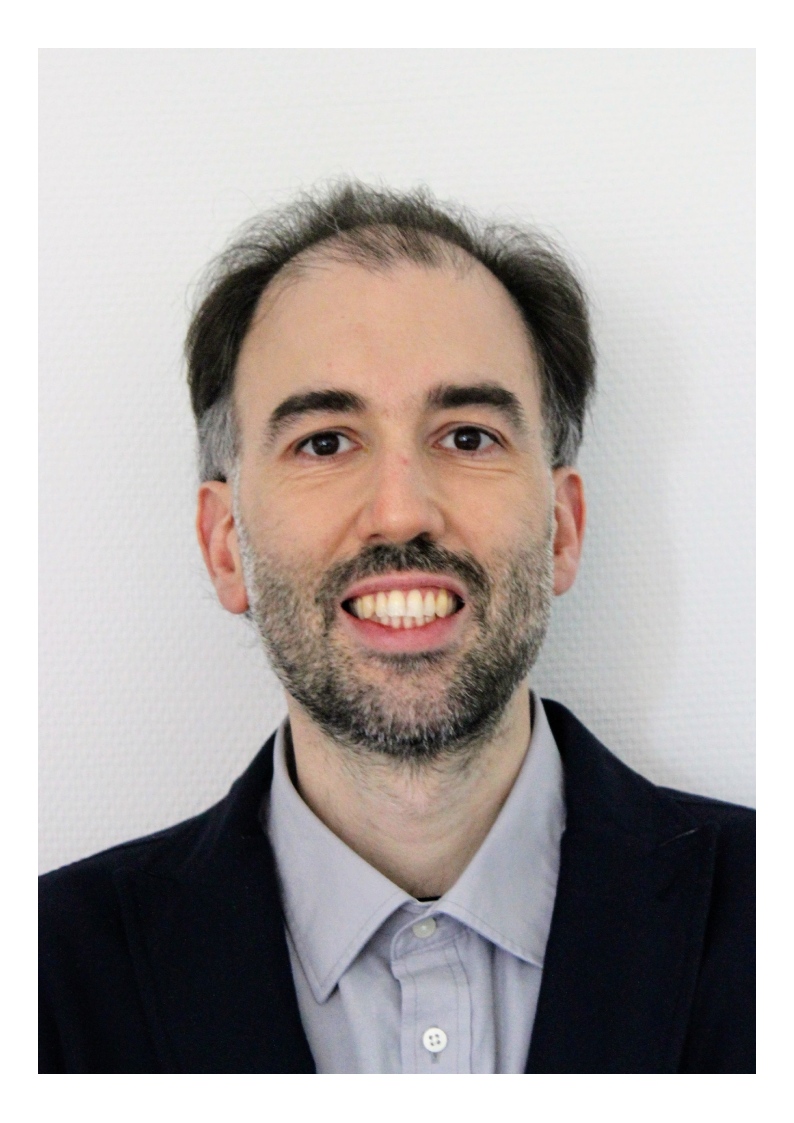}}]{Algo Car\`e}(M'19) received the Ph.D. degree in informatics and automation engineering in 2013 from the University of Brescia, Italy, where he is currently an Associate Professor with the Department of Information Engineering. After his Ph.D., he spent two years at the University of Melbourne, VIC, Australia, as a Research Fellow in system identification with the Department of Electrical and Electronic Engineering. In 2016, he was a recipient of a two-year ERCIM Fellowship that he spent at the Institute for Computer Science and Control (SZTAKI), Hungarian Academy of Sciences (MTA), Budapest, Hungary, and at the Multiscale Dynamics Group, National Research Institute for Mathematics and Computer Science (CWI), Amsterdam, The Netherlands.  He was one of the recipients of the 2025 IEEE CSS {\em Roberto Tempo} CDC Best Paper Award and he received the triennial Stochastic Programming Student Paper Prize of the Stochastic Programming Society for the period 2013--2016.
He is an Associate Editor for {\em Automatica} and for the {\em International Journal of Adaptive Control and Signal Processing}. He is a member of the {\em EUCA Conference Editorial Board}, of the {\em IFAC Technical Committee on Modeling, Identification and Signal Processing}, and of the {\em IEEE Technical Committee on Systems Identification and Adaptive Control}. His current research interests include: {\em data-driven decision methods}, {\em system identification}, and {\em learning theory}.
\end{IEEEbiography}
\vspace{-1cm}
\begin{IEEEbiography}[{\includegraphics[width=1in,height=1.25in,clip,keepaspectratio]{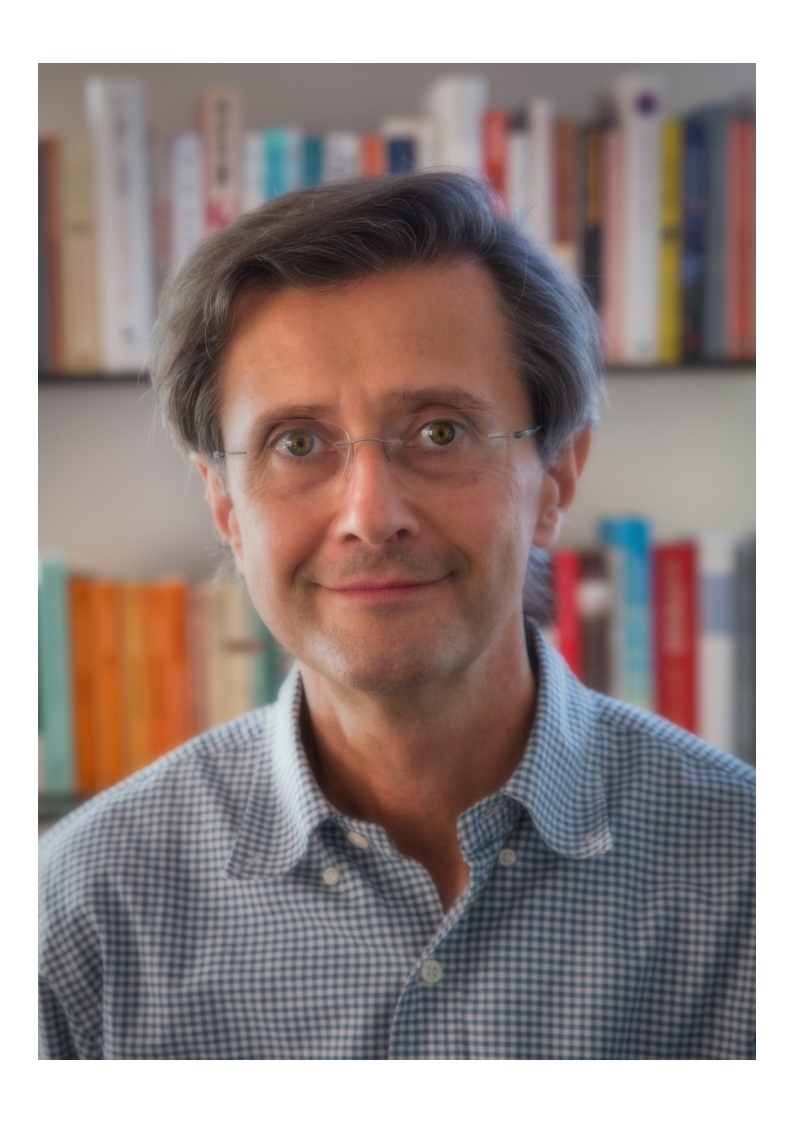}}]{Marco Claudio Campi} (F’12)  is Professor of Automatic Control at the {\em University of Brescia}, Italy. He has held visiting and teaching appointments at the {\em Australian National University}, Canberra, Australia; the {\em University of Illinois at Urbana-Champaign}, USA;  the {\em Centre for Artificial Intelligence and Robotics}, Bangalore, India; the {\em University of Melbourne}, Australia; {\em Kyoto University}, Japan; {\em Texas A{\&}M University}, USA; and the {\em NASA Langley Research Center}, Hampton, Virginia. USA. Marco Campi was the chair of the {\em Technical Committee IFAC on Modeling, Identification and Signal Processing (MISP)}, and is now leading the Italian section of ERNSI - European Research Network on System Identification. He has served in various capacities on the editorial boards of {\em Automatica}, {\em Systems and Control Letters} and the {\em European Journal of Control}. In 2008, he received the IEEE CSS George S. Axelby outstanding paper award for the article {\em The Scenario Approach to Robust Control Design}. He has delivered plenary and semi-plenary addresses at major conferences, including CDC, SYSIDand , MTNS, and has served multiple terms as a {\em Distinguished Lecturer of the Control Systems Society}. Marco Campi is a Fellow of IEEE and a Fellow of IFAC. His research interests include: {\em data-driven decision making}, {\em inductive learning}, {\em system identification}, and {\em stochastic systems}.
\end{IEEEbiography}
\vspace{-1cm}
\begin{IEEEbiography}[{\includegraphics[width=1in,height=1.25in,clip,keepaspectratio]{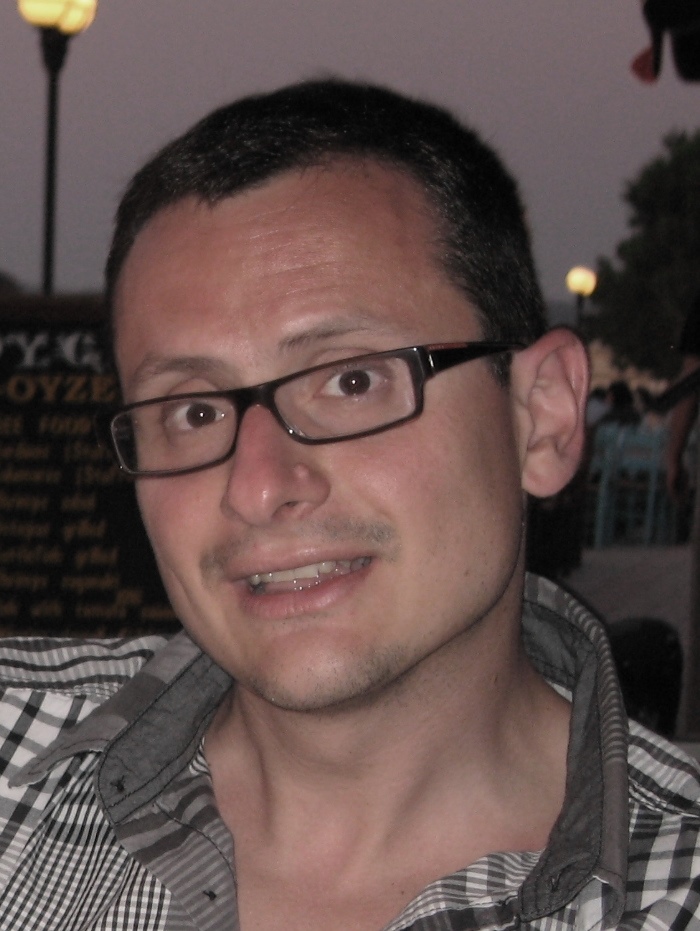}}]{Simone Garatti}(M'13) is Associate Professor at the {\em Dipartimento di Elettronica ed Informazione of the Politecnico di Milano}, Italy. He received the Laurea degree cum laude and the Ph.D. cum laude in Information Technology Engineering in 2000 and 2004, respectively, both from the Politecnico di Milano. He held visiting positions at the {\em Lund University of Technology}, at the {\em University of California San Diego}, at the {\em Massachusetts Institute of Technology}, and at the {\em University of Oxford}. Simone Garatti is currently member of the {\em IEEE-CSS Conference Editorial Board} and Associate Editor of the {International Journal of Adaptive Control and Signal Processing} and of the {\em Machine Learning and Knowledge Extraction journal}. In 2024, he served as Tutorial Chair in the organizing committee of the 6th Learning for {\em Dynamics and Control Conference} (L4DC), while he was a member of the EUCA {\em Conference Editorial Board} from 2013 to 2019. With his co-authors, Simone Garatti has pioneered the theory of the scenario approach for which he was keynote speaker at the IEEE 3rd Conference on Norbert Wiener in the 21st Century in 2021, and semi-plenary speaker at the 2022 European Conference on Stochastic Optimization and Computational Management Science (ECSO-CMS). His current research interests include: {\em data-driven optimization and decision-making}, {\em stochastic optimization}, {\em system identification}, {\em uncertainty quantification}, and {\em statistical learning theory}.
\end{IEEEbiography}

\end{document}